\documentclass[12pt]{scrartcl}
\pdfoutput=1
\usepackage{a4}
\usepackage{amsthm}
\usepackage{amsmath}
\usepackage{amssymb}
\usepackage{amsfonts}
 \usepackage{amsxtra}
\usepackage{mathrsfs}
\usepackage{latexsym}
\usepackage{color}
\usepackage{bbm}
\definecolor{Myblue}{rgb}{0,0,0.6}
\usepackage[a4paper,colorlinks,citecolor=Myblue,linkcolor=Myblue,urlcolor=Myblue,pdfpagemode=None]{hyperref}
\usepackage[square,numbers,sort&compress]{natbib}
\usepackage[all,cmtip]{xy}
\usepackage{tikz}
 \usepackage[center]{subfigure}
 \usepackage{wrapfig}
 \usepackage{layout}
 \usepackage{alltt}
 \usepackage{xypic}
 \input xy
 \xyoption{all}

  \vfuzz \hfuzz
\newcommand{\C}{\mathbbm{C}}

\newcommand{\PP}{\mathbbm{P}}

\newcommand{\R}{\mathbbm{R}}
\newcommand{\Z}{\mathbbm{Z}}
\def\1{\ifmmode\mathrm{1\!l}\else\mbox{\(\mathrm{1\!l}\)}\fi}

\newcommand{\be}{\begin{equation}}
\newcommand{\ee}{\end{equation}}
\newcommand{\bes}{\begin{equation*}}
\newcommand{\ees}{\end{equation*}}
\newcommand{\id}{\text{id}}

\newcommand{\Aa}{A_{\infty}}
\newcommand{\Ainf}{A_{\infty}}

\def\Z{{\mathbb Z}}

\def\id{{\rm id}}

\def\Wcal{{\mathcal W}}

\def\mathbi#1{\textbf{\em #1}}
\def\gbold#1{\mbox{\boldmath$#1$}}

\newcommand{\Ds}{\mathscr{D}}
\newcommand{\Es}{\mathscr{E}}

\newcommand{\Ms}{\mathscr{M}}

\newcommand{\Os}{\mathscr{O}}

\newcommand{\Ts}{\mathscr{T}}

\newcommand{\Ld}{\mathbf{L}}
\newcommand{\Rd}{\mathbf{R}}
\newcommand{\Dd}{\mathbf{D}}

\newcommand{\fl}{\rightarrow}

\newcommand\umod{\operatorname{mod \textendash \!}}

\newcommand\uEnd{\operatorname{End}}

\newcommand\uCoh{\operatorname{Coh}}
\newcommand\uExt{\operatorname{Ext}}

\newcommand\uHom{\operatorname{Hom}}

\newcommand\utr{\operatorname{tr}}

\newcommand\uU{\operatorname{U}}

\newcommand\utot{\operatorname{tot}}

\allowdisplaybreaks

\deffootnote[1em]{1em}{1em}
{\textsuperscript{\thefootnotemark}}

\theoremstyle{plain}
 \newtheorem{thm}{Theorem}[section]
 
 \newtheorem{prop}[thm]{Proposition}
 
 \theoremstyle{definition}

 \theoremstyle{remark}

\numberwithin{equation}{section}
\numberwithin{figure}{section}

\begin{document}

\title{Remarks on quiver gauge theories from open topological string theory}
\author{Nils Carqueville$^*$ \quad Alexander  Quintero V\'{e}lez$^\dagger$
\\[0.5cm]
 \normalsize{\tt \href{mailto:nils.carqueville@physik.uni-muenchen.de}{nils.carqueville@physik.uni-muenchen.de}} \quad
  \normalsize{\tt \href{mailto:a.velez@maths.gla.ac.uk}{a.velez@maths.gla.ac.uk}}\\[0.1cm]
  {\normalsize\slshape $^*$Arnold Sommerfeld Center for Theoretical Physics, }\\[-0.1cm]
  {\normalsize\slshape LMU M\"unchen, Theresienstr.~37, D-80333 M\"unchen}\\[-0.1cm]
  {\normalsize\slshape $^*$Excellence Cluster Universe, Boltzmannstr.~2, D-85748 Garching}\\[0.1cm]
  {\normalsize\slshape $^\dagger$Department of Mathematics, University of Glasgow,}\\[-0.1cm]
  {\normalsize\slshape Glasgow, G12\,8QW, United Kingdom}\\[-0.1cm]
}
\date{}
\maketitle

\vspace{-11.5cm}
\hfill {\scriptsize LMU-ASC 60/09}

\vspace{10.5cm}

\begin{abstract}
We study effective quiver gauge theories arising from a stack of D3-branes on certain Calabi-Yau singularities. Our point of view is a first principle approach via open topological string theory. This means that we construct the natural $\Aa$-structure of open string amplitudes in the associated D-brane category. Then we show that it precisely reproduces the results of the method of brane tilings, without having to resort to any effective field theory computations. In particular, we prove a general and simple formula for effective superpotentials. 
\end{abstract}

\section{Introduction and summary}\label{introduction}

In this note we wish to revisit and clarify the relation of two approaches to certain quiver gauge theories from string theory. Common to both approaches is the setting of type IIB string theory on $\R^{3,1}\times X$ where $X$ is a Calabi-Yau threefold. One then considers a stack of Minkowski space-filling D3-branes that from the point of view of $X$ are placed at a singular point obtained by shrinking a complex surface $Z\subset X$ to zero size, so that~$X$ can be identified with the total space of the canonical line bundle on~$Z$. The effective low energy field theory of this arrangement is an $\mathcal N=1$ supersymmetric quiver gauge theory~\cite{dm9603167}. 

In order to extract the data of the effective quiver gauge theory from the geometry of $Z$, several procedures have been developed in~\cite{fhh0003085, hk0503149, fhkvw0504110, hetal0505211, fhh0104259, g0807.3012}, most notably the fast inverse algorithm and the brane tiling method. These have been argued to be equivalent and hence we may restrict our attention to the more elegant latter method. As we will recall below, once a brane tiling associated to a given surface~$Z$ has been constructed, it is easy to read off the quiver gauge theory data. More precisely, one obtains the quiver that describes the field content of the effective theory as well as its F-term superpotential $W$. 

\medskip

Independently of the above developments, recent years have seen huge conceptual progress in the understanding of the sector of chiral primaries in type~II string theory~\cite{d0011, al0104, s9902, hhp0803.2045}. This sector is equivalent to the topological B-twist~\cite{w9112} of the associated sigma model with target~$X$; in particular, the effective superpotential is among the data that can be obtained from topologically twisted theories~\cite{bdlr9906200}. The D-branes and open strings of such models are described by the objects and morphisms of the bounded derived category $\Dd^b(\uCoh(X))$ of coherent sheaves on~$X$. 

In order to capture the full structure of open topological string theory on~$X$ one however needs more than just its derived category. Indeed, it was shown in~\cite{hm0006120, hll0402, c0412149} that the amplitudes
$$
W_{i_{0}\ldots i_{k}} \sim \Big\langle \psi_{i_0}\psi_{i_1}\psi_{i_2}\int\psi_{i_3}^{(1)}\ldots \int\psi_{i_{k}}^{(1)}  \Big\rangle_{\text{disk}}
$$
of open string states $\psi_i$ and their integrated descendants $\int\psi_{i}^{(1)}$ are subject to certain constraints. These constraints follow from BRST symmetry and Ward identities and precisely encode a cyclic, unital and minimal $\Aa$-structure on the D-brane category (which in our case is a subcategory of $\Dd^b(\uCoh(X))$). In terms of the amplitudes $W_{i_{0}\ldots i_{k}}$ this basically means that they have a (graded) cyclic symmetry and that they obey the $\Aa$-relations
\begin{equation}\label{AinfQ}
\sum_{\substack{r\geq 0,s\geq 1, \\ r+s\leq k}}(-1)^{|\psi_{i_1}|+\ldots+|\psi_{i_r}|+r} Q^i_{i_1\ldots i_r j i_{r+s+1}\ldots i_k} Q^j_{i_{r+1}\ldots i_{r+s}} = 0
\end{equation}
for all $k\geq 1$, where
$$
Q^{i}_{i_{0}\ldots i_{k}} = \omega^{ij}W_{j i_{0}\ldots i_{k}}
$$
is defined via the inverse $\omega^{ij}$ of the topological metric $\omega_{ij}=\langle\psi_{i}\psi_{j}\rangle_{\text{disk}}$. Following~\cite{l0507222} we may hence identify any cyclic, unital and minimal $\Aa$-algebra (or $\Aa$-category in the case of many branes) with an open topological string theory and vice versa. This means that everything that follows from this mathematical structure may be viewed as derived from first principles. 

It is a fact that the $\Aa$-structure encoded in the ``structure constants'' $Q^{i}_{i_{0}\ldots i_{k}}$ or in the amplitudes $W_{i_{0}\ldots i_{k}}$ and the topological metric can be explicitly constucted for any open topological string theory we know of. This is true not only for models with non-compact target (relevant for quiver gauge theories) but also compact theories such as sigma models and Landau-Ginzburg models~\cite{c0904.0862, l0107, addf0404, a0703279, t0107195}; because of its fundamental nature, the $\Aa$-approach can in principle be applied to any D-brane in any open topological string theory. Given this algebraic structure, one immediately obtains the effective F-term superpotential perturbatively as from the worldsheet perspective it is simply the generating function
\begin{equation}\label{Winfty}
\Wcal_{\infty} =\sum_{k \geq 2} \sum_{i_0,\ldots,i_k}  \frac{W_{i_{0}\ldots i_{k}} }{k+1}u_{i_{0}}\ldots u_{i_{k}}
\end{equation}
in terms of formal variables $u_{i}$ ``dual'' to the fields $\psi_{i}$. 

We believe the approach to open topological string theory and its effective field theories via $\Aa$-algebras to be conceptually very clear. However, from the point of view of efficiently computing superpotentials there are several other powerful techniques. In the case of compact models these include methods relying on geometric aspects and mirror symmetry, renormalisation group flows, or abstract deformation theory; for recent results in these directions see e.\,g.~\cite{ghkk0811.2996, ks0805.1013, ab0909.2245, bw0812.3397, bbg0704.2666, ahmm0901.2937, w0904.4905, mw0709.4028, w0605162, js0904.4674, js0808.0761, ahjmms0909.1842, dgjt0203173, gj0608027, ghkk0909.2025, ghkk0912.3250}. 

\medskip

\begin{figure}
\begin{align*}
&\begin{tikzpicture}[scale=0.5,baseline,inner sep=0.7mm]
\clip (0.3,0.3) rectangle (3.7,3.7);
\node (p0) at (2,2) [circle,draw=black,fill=white] {};
\node (p1) at (1,1) [circle,draw=black,fill=black] {};
\node (p2) at (3,2) [circle,draw=black,fill=black] {};
\node (p3) at (2,3) [circle,draw=black,fill=black] {};
\draw[-, very thick] (p1) -- (p2) -- (p3) -- (p1); 
\end{tikzpicture}
\begin{tikzpicture}[scale=0.5,baseline,inner sep=0.7mm]
\clip (0.3,0.3) rectangle (3.7,3.7);
\node (p0) at (2,2) [circle,draw=black,fill=white] {};
\node (p1) at (1,2) [circle,draw=black,fill=black] {};
\node (p2) at (2,1) [circle,draw=black,fill=black] {};
\node (p3) at (3,2) [circle,draw=black,fill=black] {};
\node (p4) at (2,3) [circle,draw=black,fill=black] {};
\draw[-, very thick] (p1) -- (p2) -- (p3) -- (p4) -- (p1); 
\end{tikzpicture}
\begin{tikzpicture}[scale=0.5,baseline,inner sep=0.7mm]
\clip (0.3,0.3) rectangle (3.7,3.7);
\node (p0) at (2,2) [circle,draw=black,fill=white] {};
\node (p1) at (2,1) [circle,draw=black,fill=black] {};
\node (p2) at (3,2) [circle,draw=black,fill=black] {};
\node (p3) at (2,3) [circle,draw=black,fill=black] {};
\node (p4) at (1,3) [circle,draw=black,fill=black] {};
\draw[-, very thick] (p1) -- (p2) -- (p3) -- (p4) -- (p1); 
\end{tikzpicture}
\begin{tikzpicture}[scale=0.5,baseline,inner sep=0.7mm]
\clip (0.3,0.3) rectangle (3.7,3.7);
\node (p0) at (2,2) [circle,draw=black,fill=white] {};
\node (p1) at (2,1) [circle,draw=black,fill=black] {};
\node (p2) at (1,3) [circle,draw=black,fill=black] {};
\node (p3) at (2,3) [circle,draw=black,fill=black] {};
\node (p4) at (3,3) [circle,draw=black,fill=black] {};
\draw[-, very thick] (p1) -- (p2) -- (p3) -- (p4) -- (p1); 
\end{tikzpicture}
\begin{tikzpicture}[scale=0.5,baseline,inner sep=0.7mm]
\clip (0.3,0.3) rectangle (3.7,3.7);
\node (p0) at (2,2) [circle,draw=black,fill=white] {};
\node (p1) at (2,1) [circle,draw=black,fill=black] {};
\node (p2) at (3,2) [circle,draw=black,fill=black] {};
\node (p3) at (2,3) [circle,draw=black,fill=black] {};
\node (p4) at (1,3) [circle,draw=black,fill=black] {};
\node (p5) at (1,2) [circle,draw=black,fill=black] {};
\draw[-, very thick] (p1) -- (p2) -- (p3) -- (p4) -- (p5) -- (p1); 
\end{tikzpicture}
\begin{tikzpicture}[scale=0.5,baseline,inner sep=0.7mm]
\clip (0.3,0.3) rectangle (3.7,3.7);
\node (p0) at (2,2) [circle,draw=black,fill=white] {};
\node (p1) at (2,1) [circle,draw=black,fill=black] {};
\node (p2) at (3,3) [circle,draw=black,fill=black] {};
\node (p3) at (2,3) [circle,draw=black,fill=black] {};
\node (p4) at (1,3) [circle,draw=black,fill=black] {};
\node (p5) at (1,2) [circle,draw=black,fill=black] {};
\draw[-, very thick] (p1) -- (p2) -- (p3) -- (p4) -- (p5) -- (p1); 
\end{tikzpicture}
\begin{tikzpicture}[scale=0.5,baseline,inner sep=0.7mm]
\clip (0.3,0.3) rectangle (3.7,3.7);
\node (p0) at (2,2) [circle,draw=black,fill=white] {};
\node (p1) at (2,1) [circle,draw=black,fill=black] {};
\node (p2) at (3,1) [circle,draw=black,fill=black] {};
\node (p3) at (3,2) [circle,draw=black,fill=black] {};
\node (p4) at (2,3) [circle,draw=black,fill=black] {};
\node (p5) at (1,3) [circle,draw=black,fill=black] {};
\node (p6) at (1,2) [circle,draw=black,fill=black] {};
\draw[-, very thick] (p1) -- (p2) -- (p3) -- (p4) -- (p5) -- (p6) -- (p1); 
\end{tikzpicture}
\begin{tikzpicture}[scale=0.5,baseline,inner sep=0.7mm]
\clip (0.3,0.3) rectangle (3.7,3.7);
\node (p0) at (2,2) [circle,draw=black,fill=white] {};
\node (p1) at (2,1) [circle,draw=black,fill=black] {};
\node (p2) at (3,2) [circle,draw=black,fill=black] {};
\node (p3) at (3,3) [circle,draw=black,fill=black] {};
\node (p4) at (2,3) [circle,draw=black,fill=black] {};
\node (p5) at (1,3) [circle,draw=black,fill=black] {};
\node (p6) at (1,2) [circle,draw=black,fill=black] {};
\draw[-, very thick] (p1) -- (p2) -- (p3) -- (p4) -- (p5) -- (p6) -- (p1); 
\end{tikzpicture}
\\
&\begin{tikzpicture}[scale=0.5,baseline,inner sep=0.7mm]
\clip (0.3,0.3) rectangle (3.7,4.7);
\node (p0) at (2,3) [circle,draw=black,fill=white] {};
\node (p1) at (2,2) [circle,draw=black,fill=black] {};
\node (p2) at (3,4) [circle,draw=black,fill=black] {};
\node (p3) at (2,4) [circle,draw=black,fill=black] {};
\node (p4) at (1,4) [circle,draw=black,fill=black] {};
\node (p5) at (1,3) [circle,draw=black,fill=black] {};
\node (p6) at (1,2) [circle,draw=black,fill=black] {};
\draw[-, very thick] (p1) -- (p2) -- (p3) -- (p4) -- (p5) -- (p6) -- (p1); 
\end{tikzpicture}
\begin{tikzpicture}[scale=0.5,baseline,inner sep=0.7mm]
\clip (0.3,0.3) rectangle (3.7,4.7);
\node (p0) at (2,3) [circle,draw=black,fill=white] {};
\node (p1) at (1,1) [circle,draw=black,fill=black] {};
\node (p2) at (3,4) [circle,draw=black,fill=black] {};
\node (p3) at (2,4) [circle,draw=black,fill=black] {};
\node (p4) at (1,4) [circle,draw=black,fill=black] {};
\node (p5) at (1,3) [circle,draw=black,fill=black] {};
\node (p6) at (1,2) [circle,draw=black,fill=black] {};
\draw[-, very thick] (p1) -- (p2) -- (p3) -- (p4) -- (p5) -- (p6) -- (p1); 
\end{tikzpicture}
\begin{tikzpicture}[scale=0.5,baseline,inner sep=0.7mm]
\clip (0.3,0.3) rectangle (3.7,4.7);
\node (p0) at (2,3) [circle,draw=black,fill=white] {};
\node (p1) at (1,2) [circle,draw=black,fill=black] {};
\node (p2) at (2,2) [circle,draw=black,fill=black] {};
\node (p3) at (3,3) [circle,draw=black,fill=black] {};
\node (p4) at (3,4) [circle,draw=black,fill=black] {};
\node (p5) at (2,4) [circle,draw=black,fill=black] {};
\node (p6) at (1,4) [circle,draw=black,fill=black] {};
\node (p7) at (1,3) [circle,draw=black,fill=black] {};
\draw[-, very thick] (p1) -- (p2) -- (p3) -- (p4) -- (p5) -- (p6) -- (p7) -- (p1); 
\end{tikzpicture}
\begin{tikzpicture}[scale=0.5,baseline,inner sep=0.7mm]
\clip (0.3,0.3) rectangle (3.7,4.7);
\node (p0) at (2,3) [circle,draw=black,fill=white] {};
\node (p1) at (1,1) [circle,draw=black,fill=black] {};
\node (p2) at (2,2) [circle,draw=black,fill=black] {};
\node (p3) at (3,4) [circle,draw=black,fill=black] {};
\node (p4) at (2,4) [circle,draw=black,fill=black] {};
\node (p5) at (1,4) [circle,draw=black,fill=black] {};
\node (p6) at (1,3) [circle,draw=black,fill=black] {};
\node (p7) at (1,2) [circle,draw=black,fill=black] {};
\draw[-, very thick] (p1) -- (p2) -- (p3) -- (p4) -- (p5) -- (p6) -- (p7) -- (p1); 
\end{tikzpicture}
\begin{tikzpicture}[scale=0.5,baseline,inner sep=0.7mm]
\clip (0.3,0.3) rectangle (3.7,4.7);
\node (p0) at (2,3) [circle,draw=black,fill=white] {};
\node (p1) at (1,2) [circle,draw=black,fill=black] {};
\node (p2) at (2,2) [circle,draw=black,fill=black] {};
\node (p3) at (3,2) [circle,draw=black,fill=black] {};
\node (p4) at (3,3) [circle,draw=black,fill=black] {};
\node (p5) at (3,4) [circle,draw=black,fill=black] {};
\node (p6) at (2,4) [circle,draw=black,fill=black] {};
\node (p7) at (1,4) [circle,draw=black,fill=black] {};
\node (p8) at (1,3) [circle,draw=black,fill=black] {};
\draw[-, very thick] (p1) -- (p2) -- (p3) -- (p4) -- (p5) -- (p6) -- (p7) -- (p8) -- (p1); 
\end{tikzpicture}
\begin{tikzpicture}[scale=0.5,baseline,inner sep=0.7mm]
\clip (0.3,0.3) rectangle (3.7,4.7);
\node (p0) at (2,3) [circle,draw=black,fill=white] {};
\node (p1) at (1,1) [circle,draw=black,fill=black] {};
\node (p2) at (2,2) [circle,draw=black,fill=black] {};
\node (p3) at (3,3) [circle,draw=black,fill=black] {};
\node (p4) at (3,4) [circle,draw=black,fill=black] {};
\node (p5) at (2,4) [circle,draw=black,fill=black] {};
\node (p6) at (1,4) [circle,draw=black,fill=black] {};
\node (p7) at (1,3) [circle,draw=black,fill=black] {};
\node (p8) at (1,2) [circle,draw=black,fill=black] {};
\draw[-, very thick] (p1) -- (p2) -- (p3) -- (p4) -- (p5) -- (p6) -- (p7) -- (p8) -- (p1); 
\end{tikzpicture}
\begin{tikzpicture}[scale=0.5,baseline,inner sep=0.7mm]
\clip (0.3,-0.3) rectangle (3.7,4.7);
\node (p0) at (2,3) [circle,draw=black,fill=white] {};
\node (p1) at (1,0) [circle,draw=black,fill=black] {};
\node (p2) at (2,2) [circle,draw=black,fill=black] {};
\node (p3) at (3,4) [circle,draw=black,fill=black] {};
\node (p4) at (2,4) [circle,draw=black,fill=black] {};
\node (p5) at (1,4) [circle,draw=black,fill=black] {};
\node (p6) at (1,3) [circle,draw=black,fill=black] {};
\node (p7) at (1,2) [circle,draw=black,fill=black] {};
\node (p8) at (1,1) [circle,draw=black,fill=black] {};
\draw[-, very thick] (p1) -- (p2) -- (p3) -- (p4) -- (p5) -- (p6) -- (p7) -- (p8) -- (p1); 
\end{tikzpicture}
\begin{tikzpicture}[scale=0.5,baseline,inner sep=0.7mm]
\clip (0.3,-0.3) rectangle (4.3,4.7);
\node (p0) at (2,3) [circle,draw=black,fill=white] {};
\node (p1) at (1,1) [circle,draw=black,fill=black] {};
\node (p2) at (2,2) [circle,draw=black,fill=black] {};
\node (p3) at (3,3) [circle,draw=black,fill=black] {};
\node (p4) at (4,4) [circle,draw=black,fill=black] {};
\node (p5) at (3,4) [circle,draw=black,fill=black] {};
\node (p6) at (2,4) [circle,draw=black,fill=black] {};
\node (p7) at (1,4) [circle,draw=black,fill=black] {};
\node (p8) at (1,3) [circle,draw=black,fill=black] {};
\node (p9) at (1,2) [circle,draw=black,fill=black] {};
\draw[-, very thick] (p1) -- (p2) -- (p3) -- (p4) -- (p5) -- (p6) -- (p7) -- (p8) -- (p9) -- (p1); 
\end{tikzpicture}
\end{align*}
\caption{The $16$ equivalence classes of reflexive lattice polygons.} 
\label{refpol} 
\end{figure}
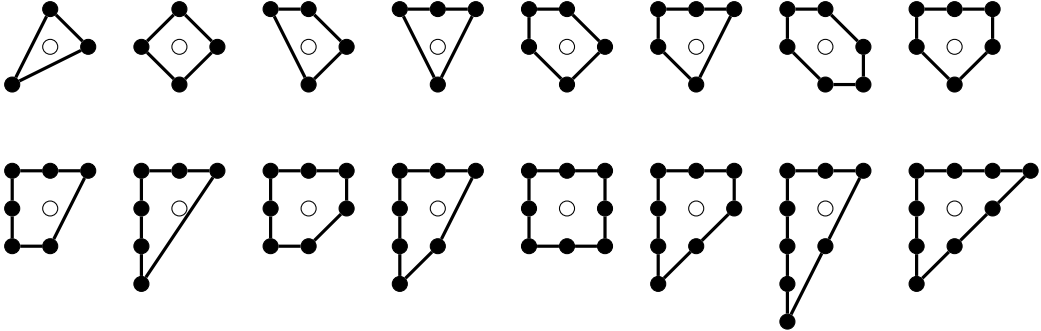

Returning to the subject of this note, it is natural to expect that the effective superpotential $\Wcal_{\infty}$ of the $\Aa$-approach coincides with the quiver gauge theory superpotential obtained by seemingly entirely different means from the geometry of the surface~$Z$ in~$X$ using the brane tiling method. This is the idea behind~\cite{af0506041} where it was checked by explicit computations that the results of the brane tiling approach and the approach via $\Aa$-algebras agree for the special two cases where the surface~$Z$ is given by either $\PP^2$ or the first del Pezzo surface $\mathrm{dP}_{1}$. 

In the present note we generalise this to all weak toric Fano surfaces~$Z$, i.\,e.~to all toric surfaces represented by reflexive polytopes. It is a classical result that there are exactly $16$ equivalence classes of reflexive polygons in the plane, shown in figure~\ref{refpol}. Since the $\Aa$-structure of a D-brane category may be viewed as the defining property of the associated open topological string theory, our result is conceptually relevant as it shows that the brane tiling method follows rigorously from first principles. 

We stress that it is not necessary to verify the above correspondence by comparing the explicit computations on both sides for each individual example of a surface. Instead, one may immediately extract an $\Aa$-structure from the superpotential $W$ of any brane tiling, and this $\Aa$-structure contains all the information about the effective quiver gauge theory. As we explain in section~\ref{sec:dimer} this is done most naturally by rephrasing the $\Aa$-relations~\eqref{AinfQ} in terms of $W$ and the topological metric, i.\,e.~the Serre pairing on the Calabi-Yau~$X$, using the language of formal non-commutative geometry~\cite{kontsevich1993} pioneered in the string theory literature in~\cite{l0507222}. 

In a second step we then show in section~\ref{sec:derivedcategory} that the $\Aa$-structure encoded in~$W$ is the \textit{same} as the natural $\Aa$-structure of branes wrapping the surface~$Z$. We denote this subsector of the open topological string theory by $\Dd^b_{Z}(\uCoh(X))$. To construct the $\Aa$-structure, we first use the elegant method of~\cite{kellerainfinrep} to find the higher products on $\Dd^b(\uCoh(Z))$ and then extend~\cite{AM0405134, af0506041, s0702539, b0801.3499, cfikv0110028, s0907.2063} it to branes in $\Dd^b_{Z}(\uCoh(X))$. Along the way we provide a general and clear proof of the explicit formula
\be\label{Wrel}
\Wcal_{\infty} = \sum_{j} r_{j} \rho_{j} 
\ee
where $r_{j}$ are the relations of the quiver associated to~$Z$, i.\,e.~polynomials of elements in $\text{Ext}^1_{Z}$ of a certain algebra describing fractional branes, and the $\rho_{j}$ are dual to $r_{j}$ so that together with the basis of $\text{Ext}^1_{Z}$ they span the space $\text{Ext}^1_{X}$ on~$X$.\footnote{The meaning of this will of course be made precise below in section~\ref{sec:derivedcategory}.} This formula was conjectured and checked in examples in~\cite{af0506041} and it is also closely related to the work of~\cite{s0702539}. 

We remark that the outcome~\eqref{Wrel} for the superpotential $\Wcal_{\infty}$ is a very simple expression. This means that while our conceptual derivation of the identity $\Wcal_{\infty}=W$ may not make the extraction of the quiver gauge theory data from a given geometry easier still, the result for $\Wcal_{\infty}$ from the $\Aa$-approach also does not rank behind its brane tiling equivalent~\eqref{Wbt} for $W$ in terms of simplicity. 

\begin{figure}
\centering
\begin{tikzpicture}[scale=1.0,baseline,implies/.style={double,double equal sign distance,-implies},dot/.style={shape=circle,fill=black,minimum size=2pt,
                inner sep=0pt,outer sep=2pt},>=stealth]

\node (p1) at (0,5) [shape=rectangle,draw=black,fill=white] {$\begin{array}{c} \text{branes on} \\ Z\subset X\end{array}$};
\node (p2) at (0,0) [shape=rectangle,draw=black,fill=white] {$\begin{array}{c} \text{quiver gauge theory} \\ (Q,W)\end{array}$};
\node (p3) at (10,5) [shape=rectangle,draw=black,fill=white] {$\begin{array}{c} \text{D-brane category} \\ \Dd^b_{Z}(\uCoh(X))\end{array}$};
\node (p4) at (5,0) [draw=white,fill=white] {$\Aa^{\text{brane tiling}}$};
\node (p5) at (10,0) [draw=white,fill=white] {$\Aa^{\text{oTST}}$};

\draw[->, very thick] (p1) -- (p2) node[midway,sloped,above] {brane tiling}; 
\draw[->, very thick] (p1) -- (p3) node[midway,sloped,above] {open topological string theory}; 
\draw[->, very thick] (p2) -- (p4) node[midway,sloped,above] {extract}; 
\draw[->, very thick] (p3) -- (p5) node[midway,sloped,above] {minimal model};  
\draw[-, double distance=2pt] (p4) -- (p5); 
 
\end{tikzpicture}
\caption{The $\Aa$-structure $\Aa^{\text{oTST}}$ of the open topological string theoretic D-brane category $\Dd^b_{Z}(\uCoh(X))$ is identical to $\Aa^{\text{brane tiling}}$, the $\Aa$-structure encoded in the brane tiling superpotential~$W$.} 
\label{outline} 
\end{figure}
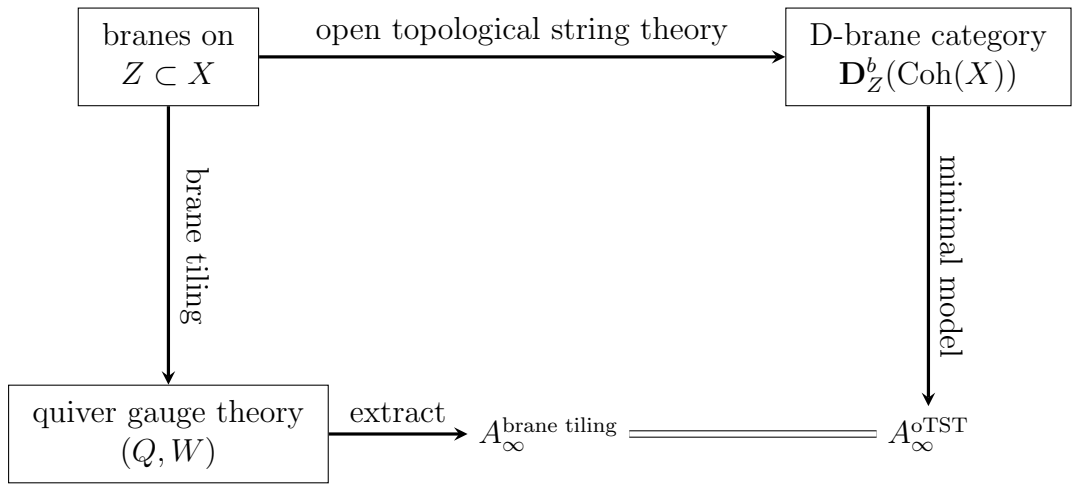

The rest of the present note is organised as follows. In sections~\ref{sec:dimer} and~\ref{sec:derivedcategory} we carry out the analysis outlined above and summarised in figure~\ref{outline}, and in section~\ref{sec:example} we offer a concrete example to complement the abstract argument of the previous two sections. The body of this note is supplemented by several appendices that collect notational conventions, mathematical background material as well as some technical details of our arguments. 

\medskip

\noindent\textbf{Note added. } While we were in the process of developing the results presented in this note into further directions, the preprint~\cite{fu0912.1656} appeared in which  a similar equivalence of $\Aa$-categories is independently proven. The main differences to our approach and proof are as follows. The authors of~\cite{fu0912.1656} present a general, yet from the string theory perspective ad hoc definition of an $\Aa$-structure associated to an arbitrary brane tiling; we on the other hand directly derive it from the superpotential using the formal non-commutative geometry of $\Aa$-algebras. Furthermore, the work of~\cite{fu0912.1656} is purely mathematical and does not discuss the relevance for string theory. We are mainly motivated by the conceptual physical question of how the brane tiling approach to quiver gauge theories can be understood from the principles of open topological string theory. Finally, our proof differs from the one of~\cite{fu0912.1656} in that we emphasise the role of the surface~$Z$ and find the $\Aa$-structure on $\Dd^b_Z(\uCoh(X))$ by extending that of $\Dd^b(\uCoh(Z))$; to construct the latter we do not have to perform explicit calculations since we make use of Keller's higher multiplication theorem~\cite{kellerainfinrep}.

\section[The $\Aa$-structure of brane tiling superpotentials]{The $\boldsymbol{\Aa}$-structure of brane tiling superpotentials}\label{sec:dimer}

Recall our general assumptions:~$Z$ is a toric surface and $X=\utot(\omega_Z)$ is the total space of the canonical bundle of $Z$. As argued in~\cite{hk0503149, fhkvw0504110, hetal0505211}, the data of the effective quiver gauge theory arising from a stack of D-branes wrapping $Z \subset X$ can be described by a certain bipartite graph on a two-torus obtained from the convex polytope of $Z$ called a brane tiling. Such a graph encodes both the quiver and the superpotential, which can be constructed in the following way. The dual graph of the brane tiling is the periodic quiver $Q$. To each vertex $i$ is associated a gauge group $\uU(n_i)$, while each arrow $a\colon i \fl j$ corresponds to a chiral superfield~$\Phi_a$ transforming in the fundamental representation of the gauge group~$\uU(n_{j})$ and the antifundamental representation of $\uU(n_{i})$. The faces of the periodic quiver are the terms in the superpotential. More precisely, the superpotential is given by
$$
\widehat{W} = \sum_{\substack{\partial f = a_{i_0}\ldots a_{i_k}}}(-1)^f \utr(\Phi_{a_{i_0}} \ldots \Phi_{a_{i_k}})
$$    
where $f$ ranges over the faces of $Q$, $\partial f =a_{i_0}\ldots a_{i_k}$ denotes the set of arrows which are contained in the boundary of~$f$, and $(-1)^f$ is $+1$ if $f$ is a white face and $-1$ if~$f$ is a black face. Equivalently, the superpotential can be interpreted entirely in terms of quiver data: instead of $\widehat{W}$ one may consider the corresponding element $W \in \C Q/[\C Q,\C Q]$. Here, $\C Q$ denotes the the path algebra of the quiver~$Q$, and modding out by the commutator space simply accounts for the fact that all terms in~$\widehat W$ are cyclically symmetric. Put differently, $W$ is a non-commutative\footnote{This is in line with the fact that there is also a reason~\cite{hll0402,l0507222} why one may also view the generating function $\mathcal W_{\infty}$ of amplitudes in~\eqref{Winfty} as a non-commutative function of its arguments~$u_{i}$: In the closed sector, the generating function of closed string amplitudes contains the same information as the amplitudes themselves. The reason is that the latter are totally (graded) symmetric. But since open string amplitudes $W_{i_{0}\ldots i_{k}}$ are only cyclically symmetric, only a non-commutative $\mathcal W_{\infty}$ can encode all information on the $W_{i_{0}\ldots i_{k}}$.} cyclic function, i.\,e.~a formal sum of oriented cycles on the quiver $Q$, and it is given by
\begin{equation}\label{Wbt}
W = \sum_{\substack{\partial f = a_{i_0}\ldots a_{i_k}}}(-1)^f a_{i_0} \ldots a_{i_k} \, .
\end{equation}

Now, as explained in detail in the next section, one should think of the arrows $a_1,\ldots,a_n$ in the quiver $Q$ as a basis of the degree $1$ part of $\uExt^{\bullet}_X(S,S)^{\vee}$, where $S$ describes the system of fractional branes that a D$3$-brane decays into when located at the zero section $Z \subset X$. Moreover, the boundary topological metrics of the system induce a non-commutative symplectic form~$\omega$ on $T (\uExt^{\bullet}_X(S,S)^{\vee})/[T (\uExt^{\bullet}_X(S,S)^{\vee}),T (\uExt^{\bullet}_X(S,S)^{\vee})]$ in the sense of \cite{kontsevich1993}, where $T (\uExt^{\bullet}_X(S,S)^{\vee})=\bigoplus_{k\geq 0}(\uExt^{\bullet}_X(S,S)^{\vee})^{\otimes k}$ (see the next section for a precise definition of the tensor product involved). Analogously to the ordinary case, one can define Poisson brackets
$$
\{ F,G \}=\frac{\partial_r F}{\partial u_i} \omega^{ij} \frac{\partial_l G}{\partial u_j}
$$
for cyclic functions~$F$ and ~$G$, where $\omega^{ij}$ is the inverse matrix to $\omega_{ij}$ and $\{u_i\}$ is a homogeneous basis for $\uExt^{\bullet}_X(S,S)^{\vee}$. The subscripts $r$ and $l$ refer to right- and left-differentiation. Viewing the superpotential $W$ as an element of $T (\uExt^{\bullet}_X(S,S)^{\vee})/[T (\uExt^{\bullet}_X(S,S)^{\vee}),T (\uExt^{\bullet}_X(S,S)^{\vee})]$ and using \eqref{Wbt}, one finds that, in fact, $\{W,W\}=0$. Indeed, the only elements appearing in $W$ have degree $1$; thus, in matrix notation the relevant part\footnote{The topological metric has ghost number $3$ and hence induces non-degenerate pairings $\uExt^1_X(S,S) \times \uExt^2_X(S,S) \fl \C$ and $\uExt^0_X(S,S) \times \uExt^3_X(S,S) \fl \C$. That being so, only the $\uExt^1$'s and the $\uExt^2$'s will play a role.} of the bracket $\{W,W \}$ can be written as
$$
\left( \begin{array}{c} \partial_r W/\partial a_1 \\
\vdots \\
\partial_r W/\partial a_n \\ \hline
0 \\
\vdots \\
0 \end{array} \right)^\text{\!\!\!T}  \left(\begin{array}{ccc} \mathbf{0} & \multicolumn{1}{|c}{}&  \begin{array}{ccc} * &   \cdots & * \\
                                                 \vdots &  & \vdots \\
                                                 * & \cdots & * \end{array}  \\ \hline
                                                 \begin{array}{ccc} * &   \cdots & * \\
                                                 \vdots &  & \vdots \\
                                                 * & \cdots & * \end{array}  & \multicolumn{1}{|c}{} & \mathbf{0} \end{array}  \right) \left( \begin{array}{c}\partial_l W/\partial a_1 \\
\vdots \\
\partial_l W/\partial a_n \\ \hline
0 \\
\vdots \\
0  \end{array} \right),
$$
which is clearly zero.

The key observation now is the simple fact that the vanishing of $\{W,W\}$ for an element $W\in T (\uExt^{\bullet}_X(S,S)^{\vee})/[T (\uExt^{\bullet}_X(S,S)^{\vee}),T (\uExt^{\bullet}_X(S,S)^{\vee})]$ is equivalent to the existence of a cyclic $A_{\infty}$-structure on $\uExt^{\bullet}_X(S,S)$. More precisely, if $W=\sum_{k\geq 2} W_{i_{0}\ldots i_{k}}  u_{i_{0}}\ldots u_{i_{k}}$ then the structure constants $Q^{i}_{i_{0}\ldots i_{k}} = \omega^{ij} W_{j i_{0}\ldots i_{k}}$ satisfy the $\Aa$-relations~\eqref{AinfQ} (for an explicit check see e.\,g.~\cite[App.~D]{l0507222}). Since the brane tiling superpotential~$W$ only depends on elements of degree 1, we thus find that it encodes a cyclic, unital and minimal $A_{\infty}$-structure on $\uExt^{\bullet}_X(S,S)$.

The main point of this note is that the above $A_{\infty}$-structure encoded in the brane tiling is not just any accidental such structure, but it is the same as the natural $A_{\infty}$-structure of the category of D-branes on~$X$ wrapping~$Z$. We will show this in the following section, thereby building the conceptual bridge between open topological string theory and quiver gauge theories. 

\section[The $A_{\infty}$-structure of the D-brane category]{The $\boldsymbol{A_{\infty}}$-structure of the D-brane category $\boldsymbol{\mathbf{D}^b_{Z}(\operatorname{Coh}(X))}$}\label{sec:derivedcategory}

In this section we present the first principle approach to quiver gauge theories associated to our class of surfaces~$Z$ via $\Aa$-algebras. We will determine the $A_{\infty}$-structure on the subsector of  open topological string theory consisting of D-branes on $X$ wrapping $Z$, and show that its superpotential $\Wcal_{\infty}$ coincides with the brane tiling superpotential. First, however, we must develop our vocabulary.

We start by recalling the definition of a tilting bundle. Let $X$ be a smooth quasi-projective variety, and let $\Dd^b(\uCoh(X))$ be the bounded derived category of coherent sheaves on $X$.  A coherent sheaf~$\Ts$ on~$X$ is called a \textit{tilting sheaf} (or, when it is locally free, a \textit{tilting bundle}) if:
\begin{enumerate}
\item[(1)]it has no higher self-extensions, i.\,e.~$\uExt_X^i(\Ts,\Ts)=0$ for all $i >0$,
\item[(2)]$\Ts$ generates the derived category $\Dd^b(\uCoh(X))$, i.\,e.~the smallest triangulated subcategory of $\Dd^b(\uCoh(X))$ closed under coproducts that contains~$\Ts$ and is closed under taking direct summands 
is equal to the whole of $\Dd^b(\uCoh(X))$.
\end{enumerate}
Now let $\Dd^b(\umod A)$ denote the bounded derived category of finitely generated right modules over the algebra $A=\uEnd(\Ts)$. Then, the functors $\Rd \uHom^{\bullet}(\Ts,-)$ and $-\otimes_A^{\Ld}\Ts$ define mutually inverse equivalences between $\Dd^b(\uCoh(X))$ and $\Dd^b(\umod A)$. For constructions of tilting bundles and their relation to derived categories we refer to \cite{MR992977, MR509388, MR939472, MR1074777, MR2384610}.

A closely related notion is that of an exceptional collection. Let~$X$ be a smooth projective variety. A coherent sheaf $E$ on~$X$ is said to be {\it exceptional} if $\uExt^k_{X}(E,E)=0$ for all $k \neq 0$, and $\uHom(E,E)=\C$. An ordered set of exceptional sheaves $(E_0,\ldots,E_n)$ is called an {\it exceptional collection} if $\uExt^k_{X}(E_j,E_i)=0$ for $j>i$ and all $k$. Such a sequence is said to be {\it strong} if it satisfies the additional condition $\uExt^k_{X}(E_j,E_i)=0$ for all $i,j$ and for $k \neq 0$. Finally, it is called {\it full} if it generates the category $\Dd^b(\uCoh(X))$.  Thus, each full strong exceptional collection defines a tilting sheaf $\Es=\bigoplus_{i=0}^nE_i$, because the endomorphism algebra of $\Es$ is a triangular algebra (cf.~\cite{K97}). Vice versa, each tilting bundle whose direct summands are line bundles gives rise to a full strong exceptional collection. 

Now let us turn to the brane tiling context. For our notation for quivers, brane tilings, perfect matchings and related things, we refer to appendix~\ref{app:tiling}. Let $A$ be an algebra defined by a quiver $Q=(Q_0,Q_1)$ with superpotential $W$ associated to some consistent brane tiling, see~\eqref{Jacalg}. Denote by $\Ms_{\theta}$ the fine moduli space of $\theta$-stable representations of $A$ of dimension vector $(1,\ldots,1)$ for a generic choice of the GIT parameter $\theta$. (See \cite{MR1315461} for some background.) An important conclusion of \cite{MR2509696} is that this moduli space is a smooth toric Calabi-Yau threefold whose fan is determined by rays corresponding to perfect matchings. Let us recall the construction of a tilting bundle on $\Ms_{\theta}$ following \cite{hhv0602041, bm0909.2013}. First some terminology is required.

A \textit{walk} $\gamma$ in $Q$ is a formal composition $a_1^{\varepsilon(1)} \ldots a_m^{\varepsilon(m)}$ where $a_1,\ldots,a_m \in Q_1$ and $\varepsilon(i) \in \{ \pm 1 \}$. If $\varepsilon(k)=1$ then $a_k$ is a forward arrow in $\gamma$; otherwise $\varepsilon(k)=-1$ and $a_k^{-1}$ is a backward arrow. Given a walk $\gamma$ we set $|\gamma|=\varepsilon(1)a_1 + \ldots + \varepsilon(m)a_m \in \Z^{Q_1}$. For a fixed perfect matching $\pi$ and a walk $\gamma$, let us write $\chi_{\pi}(\gamma)\equiv\chi_{\pi}(w(|\gamma|))$ in the notation of appendix~\ref{app:tiling}.

To each perfect matching one may pick $\theta$ as in~\cite{MR2509696} and associate a prime toric divisor on $\Ms_{\theta}$ as follows:
 $$
 P \ni \pi \mapstochar\longrightarrow D_{\pi} = \{ x \in \Ms_{\theta} \mid \text{$x_a=0$ for $a \in \pi$}\} \, .
 $$
Here by $a\in\pi$ we mean that the arrow~$a$ crosses the perfect matching~$\pi$. If $D_{\pi}$ is a compact toric divisor then the perfect matching $\pi$ is said to be {\it internal}. For any walk $\gamma$, one can define a toric divisor on $\Ms_{\theta}$ by setting $D(\gamma)=\sum_{\pi \in P}\chi_{\pi}(\gamma)D_{\pi}$. We denote by $\Os(\gamma)\equiv\Os(D(\gamma))$ the sheaf corresponding to $D(\gamma)$. 

Now fix a vertex $i_0 \in Q_0$; for any other vertex $i \in Q_0$, choose an arbitrary walk~$\gamma_i$ from~$i_0$ to~$i$. Then there is the following important result.

\begin{thm}[\cite{hhv0602041, bm0909.2013}]\label{thm:tilting}
Let the notation and assumptions be as above. Then $\Ts = \bigoplus_{i \in Q_0}\Os(\gamma_i)$ is a tilting bundle on $\Ms_{\theta}$. In addition, $A \cong \uEnd(\Ts)$.
\end{thm}

Now let us go back to the specific context of this section.  Let~$Z$ be a weak toric Fano surface, and denote by $X=\utot(\omega_Z)$ the total space of its canonical bundle. As discussed in~\cite{iu0905.0059}, there is a consistent brane tiling corresponding to~$Z$ (see also~\cite{MR2454323, g0807.3012}). Furthermore, the arguments of \cite{MR2509696} imply that $\Ms_{\theta}$ coincides with $X$ for a particular choice of the GIT parameter $\theta$. In light of the above discussion, we deduce that $\Ts = \bigoplus_{i \in Q_0}\Os(\gamma_i)$ is a tilting bundle on $X$ whose direct summands are line bundles. 

To take the next step, we should construct a full strong exceptional collection on $Z$. To this end, let $\iota \colon Z \hookrightarrow  X$ be the inclusion of the zero section. Then it is shown in appendix~\ref{app:tilting} that the restriction $\Es = \iota^* \Ts$ is a tilting bundle on $Z$ whose summands are line bundles. This, in turn, gives rise to a full strong exceptional collection on $Z$. In what follows we denote by $B$ the algebra of endomorphisms of the object $\Es$, i.\,e.~$B=\uEnd(\Es)$.  

We observe next that the algebra $B$ can be naturally viewed as the path algebra of a quiver with relations. In order to establish this, we fix an internal perfect matching $\pi$. Let $Q(\pi) \subset Q$ be the acyclic subquiver obtained by deleting from $Q$ all the arrows corresponding to $\pi$. We regard the path algebra $\C Q(\pi)$ as a subalgebra of $\C Q$. In the notation of appendix~\ref{app:tiling}, let $J_{\pi}$ be the ideal of relations in $\C Q(\pi)$ generated by the Jacobi ideal $J_W$; that is, $J_{\pi}=\C Q(\pi) \cap J_W$. As is apparent from the definition of internal perfect matching, the prime toric divisor $D_{\pi}$ is identified with the image of the zero section $\iota \colon Z \hookrightarrow X$. Furthermore, there is a canonical surjective map $\iota^{\natural}\colon \Gamma(X,\Ts^{\vee} \otimes \Ts) \fl \Gamma(Z, \Es^{\vee} \otimes \Es)$ whose kernel is generated by sections of $\Ts^{\vee} \otimes \Ts$ vanishing at $D_{\pi}$. Observe also that the map sending a path $p=a_1 \ldots a_{m}$ in $Q(\pi)$ to the product of the corresponding sections $s_1\ldots s_m \in \uHom(\Os(\gamma_{t(p)}) , \Os(\gamma_{h(p)}))$ determines an algebra homomorphism $\eta \colon \C Q(\pi) \fl \uEnd(\Ts)$. From this, one readily deduces that $\eta$ sends paths $p,p'$ in $Q_{\pi}$ satisfying $t(p)=t(p')$ and $h(p)=h(p')$ to the same element in $ \uHom(\Os(\gamma_{t(p)}) , \Os(\gamma_{h(p)}))$ vanishing at $D_{\pi}$ if and only if $p-p' \in J_{\pi}$. Therefore
$$
B=\uEnd(\Es) \cong \C Q(\pi)/J_{\pi} \, ,
$$
which proves the initial assertion.

We now turn our attention to obtaining some important consequences of the preceding results. Denote by $\umod B$ the category of finitely generated right modules over $B$. For each vertex $i \in Q_0$ we have a simple object $T_i$ in $\umod B$; this is the representation which assigns the field $\C$ to the vertex $i$ and $0$ to any other vertex and where each arrow gives the zero map. We obtain in this way a complete set of representatives for the isomorphism classes of simple $B$-modules. Let $T$ be the sum $\bigoplus_{i \in Q_0} T_i$ of all simple modules. The chain complex of endomorphisms $\Rd \uHom_B^{\bullet}(T,T)$ has a natural structure of a differential graded algebra. Multiplication is given by composition of endomorphisms. As explained in appendix~\ref{app:Ainfinity}, there is an $A_{\infty}$-structure on the full Ext-algebra 
$$
V^{\bullet} = \uExt_B^{\bullet}(T,T)
$$
with $m_1=0$ and $m_2$ is induced by the multiplication of $\Rd \uHom_B^{\bullet}(T,T)$. This $A_{\infty}$-structure on $V^{\bullet}$ is unique up to $A_{\infty}$-isomorphisms, and we will determine it explicitly below. 

As follows from the remarks at the beginning of this section, the category $\Dd^b(\uCoh(Z))$ is equivalent to $\Dd^b(\umod B)$. This equivalence is given by the functor $\Rd \uHom^{\bullet}(\Es,-)$. By a general result of \cite[Thm.~3.1]{MR2258042}, the triangulated subcategory of $\Dd^b(\umod B)$ generated by the $T_i$ is equivalent to the derived category $\Dd^b(V^{\bullet})$ (obtained as the degree zero cohomology of the $A_{\infty}$-category of ``twisted complexes'' over $V^{\bullet}$, see e.\,g.~\cite{l0310337, l0610120}). Since $T_i$, $i \in Q_0$, generate the derived category $\Dd^b(\umod B)$, we obtain an equivalence between $\Dd^b(\umod B)$ and $\Dd^b(V^{\bullet})$.

To summarize, we have proved the following.

\begin{prop}
The derived category of coherent sheaves $\Dd^b(\uCoh(Z))$ on a weak Fano surface $Z$ is equivalent to the derived category $\Dd^b(V^{\bullet})$.
\end{prop}

Now we consider the algebra $A$ corresponding to $X$. Keeping the above notation, the category of finitely generated right modules over $A$ will be denoted by $\umod A$. Again, to every vertex $i \in Q_0$ corresponds canonically a simple object~$S_i$ in $\umod A$. It is easily verified that the~$S_i$ coincide with the objects $\iota_*T_i$. 
In contrast to the above, these will not be the only simple modules. The importance of the $S_i$ lies in the fact that they are the fractional branes which give rise to the effective quiver gauge theory. We can again form the sum $S= \bigoplus_{i \in Q_0} S_i$ and the $A_{\infty}$-algebra 
$$
\Lambda^{\bullet}=\uExt^{\bullet}_A(S,S) \, ,
$$
describing open strings stretched between the fractional branes. There is the following connection between this new $A_{\infty}$-algebra and the previous one (again, see appendix~\ref{app:Ainfinity} for $\Aa$-terminology).

\begin{thm}[\cite{af0506041, s0702539, b0801.3499, cfikv0110028}]
$\Lambda^{\bullet}$ is $A_{\infty}$-quasi-isomorphic to the trivial extension $V^{\bullet} \oplus V^{\bullet} [-3]^{\vee}$.
\end{thm}

We should take a moment here to make one important point. Let $\Dd_0^b(\umod A)$ be the full triangulated subcategory of $\Dd^b(\umod A)$ generated by the simple modules $S_i$.  Arguing as before, we see that $\Dd_0^b(\umod A)$ is equivalent to the derived category $\Dd^b(\Lambda^{\bullet})$. On the other hand, define $\Dd_Z^b(\uCoh(X)) \subset \Dd^b(\uCoh(X))$ to be the full subcategory consisting of objects all of whose cohomology sheaves are supported on the zero section $Z \subset X$. A result of \cite[Lem.~4.4]{MR2142382} provides an equivalence between $\Dd^b_Z(\uCoh(X))$ and $\Dd^b_0(\umod A)$. The net conclusion is that there is an equivalence 
$$
\Dd^b_Z(\uCoh(X)) \cong \Dd^b(\Lambda^{\bullet}) \, .
$$
This means that  the structure of the open topological string theory described by the D-brane category $\Dd^b_Z(\uCoh(X))$ is controlled by the relatively simple $A_{\infty}$-algebra $\Lambda^{\bullet}$ of open strings between fractional branes.

Now we arrive at a key junction. Let $E_0=\bigoplus_{i \in Q_0} \C \cdot e_i$ be the semisimple commutative subalgebra of $\C Q$ spanned by the trivial paths. We define $E_1=\bigoplus_{a \in Q_1}\C \cdot a$ with the natural structure of an $E_0$-bimodule. For each arrow $a \in Q_1$, there is a relation $r_a=\partial_a W$ which is the cyclic derivative of the superpotential~$W$. Let $R=\bigoplus_{a \in Q_1} \C \cdot r_a$ be the $E_0$-bimodule generated by these. Then the path algebra of $Q$ is $\C Q=T_{E_0}E_1=\bigoplus_{k \geq 0}T^k_{E_0}E_1$, where $T^k_{E_0}E_1=E_1 \otimes_{E_0} \ldots \otimes_{E_0}E_1$ is the $k$-fold tensor product. We learn, at the same time, that the algebra $A$ is the quotient of $T_{E_0}E_1$ by the $2$-sided ideal generated by the $E_0$-bimodule $R$. We also have the following canonical isomorphisms (cf.~\cite{MR2533303})
\begin{align*}
E_1^{\vee} &\cong \uExt_A^1(S,S) = \uExt_B^1(T,T) \oplus \uExt_B^2(T,T)^{\vee} \, ,\\
R^{\vee} &\cong \uExt_A^2(S,S)=\uExt_B^2(T,T) \oplus \uExt_B^1(T,T)^{\vee} \, .
\end{align*}
In the same vein, one defines $E_0$-bimodules $E_1(\pi)= \bigoplus_{a \in Q(\pi)_1}\C \cdot a$ and $R(\pi)=\bigoplus_{a \in \pi} \C \cdot r_a$, where we recall that $a\in\pi$ means that the arrow~$a$ crosses the perfect matching~$\pi$, i.\,e.~$a\in Q_{1}\backslash Q(\pi)_{1}$.
Once again, we can represent $B$ as the quotient of the tensor algebra $T_{E_0}E_1(\pi)$ by the $2$-sided ideal generated by the space of relations $R(\pi)$. Moreover, there are canonical isomorphisms
$$
E_1(\pi)^{\vee} \cong \uExt_B^1(T,T) \quad \text{and} \quad R(\pi)^{\vee} \cong \uExt_B^2(T,T) \, .
$$

The following important result is a version of the minimal model theorem discussed in appendix~\ref{app:Ainfinity} that is specially adapted to the case of quivers. It will enable us to take full advantage of the results so far obtained. 

\begin{thm}[Keller's higher multiplication theorem~\cite{kellerainfinrep}]\label{higherKeller}
Let $i \colon R(\pi) \hookrightarrow T_{E_0}E_1(\pi)$ be the inclusion map and let $\lambda$ be the composite $R \xrightarrow{i} T_{E_0}E_1(\pi) \fl \bigoplus_{k \geq 0}E_1(\pi)^{\otimes k}$. Then the only non-vanishing higher multiplications $m_k$ of $\Lambda^\bullet=\uExt_B^{\bullet}(S,S)$ are equal to the maps
$$
\lambda_k^{\vee}\colon \uExt_B^1(T,T)^{\otimes k} \longrightarrow \uExt_B^2(T,T) \, .
$$
Thus essentially one knows these multiplications as soon as one knows all the relations in the quiver. 
\end{thm}

With all this in place, we are now ready to compare the superpotential extracted from the $A_{\infty}$-structure inherent in the derived category of coherent sheaves on $X$ and the superpotential extracted from the brane tiling. Let us denote by $\{a_1,\ldots,a_n \}$ the set of arrows in $Q(\pi)$ and consider it as a basis in $E_1(\pi)$. Denote by $\{\alpha_1,\ldots,\alpha_n\}$ the basis for $\uExt^1_B(T,T)$ which is dual to $\{a_1,\ldots,a_m\}$. Let $\{ r_1,\ldots,r_m\}$ be the full set of relations in $Q(\pi)$, where each $r_j$ can be written
$$
r_j =\sum_{k \geq 2}\sum_{i_1,\ldots,i_k}(a_{i_1}\ldots a_{i_k},r_j) a_{i_1} \otimes \ldots \otimes a_{i_k} \, ,
$$
for some $(a_{i_1}\ldots a_{i_k},r_j) \in \C$. We write $\{\rho_1,\ldots,\rho_m \}$ for the generators of $\uExt^2_B(T,T)$ dual to $\{ r_1,\ldots,r_m\}$. By theorem~\ref{higherKeller}, the higher multiplications $m_k$ are ``dual to the relations'':
$$
m_k(\alpha_{i_1},\ldots,\alpha_{i_k})=\sum_j (a_{i_1}\ldots a_{i_k},r_j) \rho_j \, .
$$
Now consider the basis $\{ \mathbi{a}_1,\ldots,\mathbi{a}_n, \gbold{\rho}_1,\ldots,\gbold{\rho}_m\}$ of $E_1$ given by $\mathbi{a}_i=(a_i,0)$ and $\gbold{\rho}_j=(0,\rho_j)$. Then $\uExt_A^1(S,S)$ has a basis $\{\gbold{\alpha}_1,\ldots,\gbold{\alpha}_n,\mathbi{r}_1,\ldots,\mathbi{r}_m \}$, where $\gbold{\alpha}_i=(\alpha_i,0)$ and $\mathbi{r}_j=(0,r_j)$. In the same manner, we obtain a basis $\{\gbold{\alpha}^*_1,\ldots,\gbold{\alpha}^*_n,\mathbi{r}^*_1,\ldots,\mathbi{r}^*_m \}$ of $\uExt_A^2(S,S)$ with $\gbold{\alpha}^*_i=(0,a_i)$ and $\mathbi{r}^*_j=(\rho_j,0)$. We also note that the $A_{\infty}$-algebra $\Lambda^{\bullet}=\uExt^{\bullet}_A(S,S)$ carries a canonical Calabi-Yau pairing given by
$$
\langle \gbold{\alpha}^*_i,\gbold{\alpha}_j \rangle = \langle \mathbi{r}^*_i,\mathbi{r}_j \rangle = \delta_{ij} \, .
$$
This pairing coincides with the topological metric. Combining our remarks with the results in appendix~\ref{app:Ainfinity}, we have for the higher multiplications $\mathbi{m}_k$ of $\Lambda^{\bullet}$:
$$
\mathbi{m}_k(\gbold{\alpha}_{i_1},\ldots,\gbold{\alpha}_{i_k})=\sum_j (a_{i_1}\ldots a_{i_k},r_j) \mathbi{r}^*_j \, .
$$
Together with the cyclicity property this determines all the $A_{\infty}$-products in $\Lambda^{\bullet}$. 

We have now accumulated all the information necessary to compute the superpotential of the low energy effective field theory of the open
topological string theory described by $\Dd^b_Z(\uCoh(X))$; we find
\begin{align*}
\Wcal_{\infty}&=\sum_{k \geq 2} \sum_{i_0,\ldots,i_k}\langle \mathbi{m}_k(\gbold{\alpha}_{i_1},\ldots,\gbold{\alpha}_{i_k}), \mathbi{r}_{i_0} \rangle \gbold{\rho}_{i_0} \otimes \mathbi{a}_{i_1}\otimes\ldots \otimes\mathbi{a}_{i_k} \\
&= \sum_{k \geq 2} \sum_{i_1,\ldots,i_k} \Bigg(\sum_{i,j} (a_{i_1}\ldots a_{i_k},r_i)\langle \mathbi{r}^*_i,\mathbi{r}_j \rangle \gbold{\rho}_{j}\Bigg) \otimes \mathbi{a}_{i_1}\otimes\ldots \otimes\mathbi{a}_{i_k}\\
&= \sum_j \Bigg(   \sum_{k \geq 2} \sum_{i_1,\ldots,i_k} (a_{i_1}\ldots a_{i_k},r_j)\mathbi{a}_{i_1}\otimes\ldots \otimes\mathbi{a}_{i_k}  \Bigg) \otimes \gbold{\rho}_{j} \\
&= \sum_j \mathbi{r}_j \otimes \gbold{\rho}_{j} \, .
\end{align*}
This result was argued for heuristically in \cite[Sect.~4]{af0506041}.

To conclude, we will show that the brane tiling superpotential $W$ is the same as~$\Wcal_{\infty}$. This turns out to be very straightforward. To begin with, each arrow $a \in Q_1$ occurs twice and only twice in the superpotential~$W$, and from~\eqref{Wbt} we see that two terms involving any given arrow appear with opposite signs. Furthermore, any pair of terms may have only one single arrow in common. Thus, using the notation introduced in the text preceding theorem~\ref{higherKeller}, we deduce that~$W$ can be expressed in the form
$$
W=\sum_{a \in \pi} r_a a \, ,
$$ 
where we recall that the construction fixes an arbitrary internal perfect matching~$\pi$. Since the arrows crossing~$\pi$ correspond to relations in the quiver $Q(\pi)$, we obtain the desired assertion. 

Altogether, we have established the following result.
\medskip
\begin{center}
\parbox[b]{0.95\textwidth}{\sl The $A_{\infty}$-structure of the D-brane category $\Dd^b_Z(\uCoh(X))$ is the same as the one encoded in the brane tiling superpotential $W$. In particular, the effective quiver gauge theory can be obtained solely from the fundamental $A_{\infty}$-structure of $\Dd^b_Z(\uCoh(X))$.}
\end{center}

\section[An example:~$\mathrm{dP}_{2}$]{An example:~$\boldsymbol{\mathrm{dP}_{2}}$}\label{sec:example}

Let us illustrate how the general formalism of the previous section applies in the concrete case where~$Z$ is $\PP^2$ blown up at two points, also known as the second del Pezzo surface $\mathrm{dP}_2$. We will explicitly construct a full strong exceptional collection on $\mathrm{dP}_2$ and determine the associated acyclic quiver and its relations. This allows us to apply equation~\eqref{Wrel} to compute the effective superpotential~$\mathcal W_{\infty}$, and of course we find that it can be identified with the brane tiling superpotential~$W$ from~\eqref{Wbt}. 

The surface $\mathrm{dP}_2$ is the toric variety defined by the polytope with vertices $u_1=(1,0)$, $u_2=(0,1)$, $u_3=(-1,1)$, $u_4=(-1,0)$ and $u_5=(0,-1)$. The toric diagram for the total space of the canonical bundle of $\mathrm{dP}_2$ is given by the fan with one-dimensional cones generated by $v_0=(0,0,1)$, $v_1=(1,0,1)$, $v_2=(0,1,1)$, $v_3=(-1,1,1)$, $v_4=(-1,0,1)$ and $v_5=(0,-1,1)$, and three-dimensional cones $\sigma_i=\R_{\geq 0} v_0+\R_{\geq 0}v_i+\R_{\geq 0}v_{i+1}$ 
for $1 \leq i \leq 4$:
$$
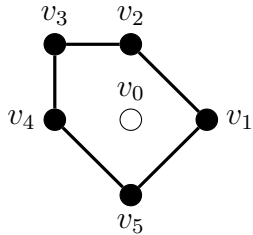
\begin{tikzpicture}[scale=1.0,baseline,inner sep=1mm]

\node (p0) at (0,0) [circle,draw=black,fill=white,label=above:$v_{0}$] {};
\node (p1) at (1,0) [circle,draw=black,fill=black,label=right:$v_{1}$] {};
\node (p2) at (0,1) [circle,draw=black,fill=black,label=above:$v_{2}$] {};
\node (p3) at (-1,1) [circle,draw=black,fill=black,label=above:$v_{3}$] {};
\node (p4) at (-1,0) [circle,draw=black,fill=black,label=left:$v_{4}$] {};
\node (p5) at (0,-1) [circle,draw=black,fill=black,label=below:$v_{5}$] {};

\draw[-, very thick] (p1) -- (p2) -- (p3) -- (p4) -- (p5) -- (p1); 
\end{tikzpicture}
$$

The associated brane tiling is
$$
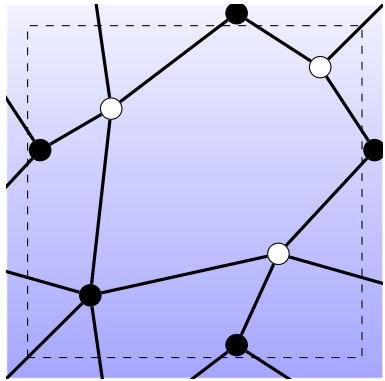
\begin{tikzpicture}[scale=.55,baseline=2cm, inner sep=1mm]
\clip (-0.5,-0.5) rectangle (8.5,8.5);
\shadedraw[top color=blue!5, bottom color=blue!35, draw= white] (-0.5,-0.5) rectangle (8.5,8.5);

\draw[dashed] (0,0) rectangle (8,8);

\node (b1) at (0.3,5) [circle,draw=black,fill=black] {};
\node (b1s) at (8.3,5) [circle,draw=black,fill=black] {};
\node (b2) at (1.5,1.5) [circle,draw=black,fill=black] {};
\node (b2s) at (9.5,1.5) [circle,draw=black,fill=black] {};
\node (b2ss) at (9.5,9.5) [circle,draw=black,fill=black] {};
\node (b2sss) at (1.5,9.5) [circle,draw=black,fill=black] {};
\node (w1) at (2,6) [circle,draw=black,fill=white] {};
\node (w1ss) at (2,-2) [circle,draw=black,fill=white] {};
\node (w1s) at (10,6) [circle,draw=black,fill=white] {};
\node (b3) at (5,0.3) [circle,draw=black,fill=black] {};
\node (b3s) at (5,8.3) [circle,draw=black,fill=black] {};
\node (w2) at (6,2.5) [circle,draw=black,fill=white] {};
\node (w2s) at (6,10.5) [circle,draw=black,fill=white] {};
\node (w2ss) at (-2,2.5) [circle,draw=black,fill=white] {};
\node (w3) at (7,7) [circle,draw=black,fill=white] {};
\node (w3s) at (-1,7) [circle,draw=black,fill=white] {};
\node (w3ss) at (-1,-1) [circle,draw=black,fill=white] {};
\node (w3sss) at (7,-1) [circle,draw=black,fill=white] {};

\draw[-, very thick] (b2) -- (w1); 
\draw[-, very thick] (b2) -- (w2); 
\draw[-, very thick] (w2) -- (b3); 
\draw[-, very thick] (w2) -- (b1s); 
\draw[-, very thick] (w3) -- (b1s); 
\draw[-, very thick] (w3) -- (b3s); 
\draw[-, very thick] (w1) -- (b3s); 
\draw[-, very thick] (w1) -- (b1); 
\draw[-, very thick] (w2) -- (b2s); 
\draw[-, very thick] (w1s) -- (b1s); 
\draw[-, very thick] (w3) -- (b2ss); 
\draw[-, very thick] (w2s) -- (b3s); 
\draw[-, very thick] (w1) -- (b2sss); 
\draw[-, very thick] (b1) -- (w3s); 
\draw[-, very thick] (b1) -- (w2ss); 
\draw[-, very thick] (b2) -- (w2ss); 
\draw[-, very thick] (b2) -- (w1ss); 
\draw[-, very thick] (b2) -- (w3ss); 
\draw[-, very thick] (b3) -- (w1ss); 
\draw[-, very thick] (b3) -- (w3sss); 

\end{tikzpicture}
$$
where the dash-bounded region is a fundamental domain of the torus. From this one may read off the dual periodic quiver~$Q$ which is given by%
$$
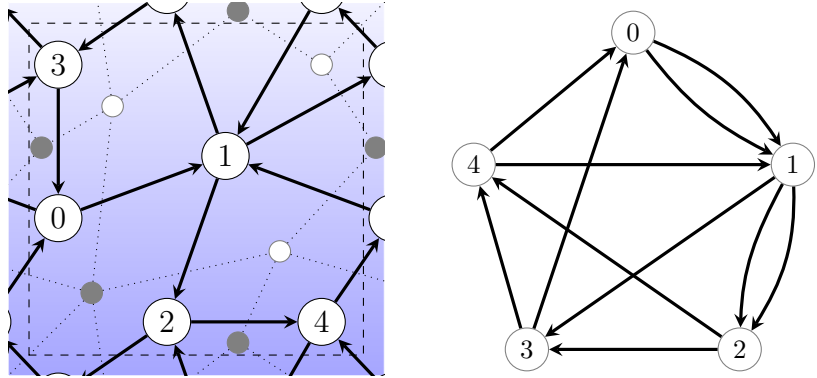
\begin{tikzpicture}[scale=.55, inner sep=1mm, >=stealth]%
\clip (-0.5,-0.5) rectangle (8.5,8.5);
\shadedraw[top color=blue!5, bottom color=blue!35, draw= white] (-0.5,-0.5) rectangle (8.5,8.5);

\draw[dashed] (0,0) rectangle (8,8);

\node (b1) at (0.3,5) [circle,draw=gray,fill= gray] {};
\node (b1s) at (8.3,5) [circle,draw= gray,fill= gray] {};
\node (b2) at (1.5,1.5) [circle,draw= gray,fill= gray] {};
\node (b2s) at (9.5,1.5) [circle,draw= gray,fill= gray] {};
\node (b2ss) at (9.5,9.5) [circle,draw= gray,fill= gray] {};
\node (b2sss) at (1.5,9.5) [circle,draw= gray,fill= gray] {};
\node (w1) at (2,6) [circle,draw= gray,fill=white] {};
\node (w1ss) at (2,-2) [circle,draw= gray,fill=white] {};
\node (w1s) at (10,6) [circle,draw= gray,fill=white] {};
\node (b3) at (5,0.3) [circle,draw= gray,fill= gray] {};
\node (b3s) at (5,8.3) [circle,draw= gray,fill= gray] {};
\node (w2) at (6,2.5) [circle,draw= gray,fill=white] {};
\node (w2s) at (6,10.5) [circle,draw= gray,fill=white] {};
\node (w2ss) at (-2,2.5) [circle,draw= gray,fill=white] {};
\node (w3) at (7,7) [circle,draw= gray,fill=white] {};
\node (w3s) at (-1,7) [circle,draw= gray,fill=white] {};
\node (w3ss) at (-1,-1) [circle,draw= gray,fill=white] {};
\node (w3sss) at (7,-1) [circle,draw= gray,fill=white] {};

\node (0) at (0.7,7) [circle,draw=black,fill=white] {3};
\node (0u) at (0.7,-1) [circle,draw=black,fill=white] {3};
\node (0r) at (8.7,7) [circle,draw=black,fill=white] {3};
\node (0ur) at (8.7,-1) [circle,draw=black,fill=white] {3};

\node (1) at (0.7,3.3) [circle,draw=black,fill=white] {0};
\node (1r) at (8.7,3.3) [circle,draw=black,fill=white] {0};

\node (2) at (7,0.8) [circle,draw=black,fill=white] {4};
\node (2o) at (7,8.8) [circle,draw=black,fill=white] {4};
\node (2l) at (-1,0.8) [circle,draw=black,fill=white] {4};
\node (2ol) at (-1,8.8) [circle,draw=black,fill=white] {4};

\node (3) at (3.3,0.8) [circle,draw=black,fill=white] {2};
\node (3o) at (3.3,8.8) [circle,draw=black,fill=white] {2};

\node (4) at (4.7,4.8) [circle,draw=black,fill=white] {1};
\node (4u) at (4.7,-3.2) [circle,draw=black,fill=white] {1};
\node (4l) at (-3.3,4.8) [circle,draw=black,fill=white] {1};

\draw[->, very thick] (0) -- (1); 
\draw[->, very thick] (1) -- (4); 
\draw[->, very thick] (4) -- (3); 
\draw[->, very thick] (3) -- (2); 
\draw[->, very thick] (1r) -- (4); 
\draw[->, very thick] (2o) -- (4); 
\draw[->, very thick] (2) -- (1r); 
\draw[->, very thick] (4) -- (0r); 
\draw[->, very thick] (4) -- (3o); 
\draw[->, very thick] (4u) -- (3); 
\draw[->, very thick] (3) -- (0u); 
\draw[->, very thick] (2) -- (4u); 
\draw[->, very thick] (0ur) -- (2); 
\draw[->, very thick] (1) -- (4l); 
\draw[->, very thick] (2l) -- (1); 
\draw[->, very thick] (3o) -- (0); 
\draw[->, very thick] (4l) -- (0); 
\draw[->, very thick] (0) -- (2ol); 
\draw[->, very thick] (0u) -- (2l);
\draw[->, very thick] (0r) -- (2o);

\draw[-, dotted] (b2) -- (w1); 
\draw[-, dotted] (b2) -- (w2); 
\draw[-, dotted] (w2) -- (b3); 
\draw[-, dotted] (w2) -- (b1s); 
\draw[-, dotted] (w3) -- (b1s); 
\draw[-, dotted] (w3) -- (b3s); 
\draw[-, dotted] (w1) -- (b3s); 
\draw[-, dotted] (w1) -- (b1); 
\draw[-, dotted] (w2) -- (b2s); 
\draw[-, dotted] (w1s) -- (b1s); 
\draw[-, dotted] (w3) -- (b2ss); 
\draw[-, dotted] (w2s) -- (b3s); 
\draw[-, dotted] (w1) -- (b2sss); 
\draw[-, dotted] (b1) -- (w3s); 
\draw[-, dotted] (b1) -- (w2ss); 
\draw[-, dotted] (b2) -- (w2ss); 
\draw[-, dotted] (b2) -- (w1ss); 
\draw[-, dotted] (b2) -- (w3ss); 
\draw[-, dotted] (b3) -- (w1ss); 
\draw[-, dotted] (b3) -- (w3sss); 

\end{tikzpicture}
\qquad
\begin{tikzpicture}[scale=.7, inner sep=1mm, >=stealth]
\clip (0.5,-0.5) rectangle (7.5,7.5);
\node (3) at (2,0) [circle,draw=gray,fill= white] {\footnotesize 3};
\node (2) at (6,0) [circle,draw=gray,fill= white] {\footnotesize 2};
\node (4) at (1,3.5) [circle,draw=gray,fill= white] {\footnotesize 4};
\node (1) at (7,3.5) [circle,draw=gray,fill= white] {\footnotesize 1};
\node (0) at (4,6) [circle,draw=gray,fill= white] {\footnotesize 0};

\draw[->, very thick] (0) to [out=-23,in=123] (1); 
\draw[->, very thick] (0) to [out=-53,in=153] (1);  
\draw[->, very thick] (1) to [out=-88,in=58] (2); 
\draw[->, very thick] (1) to [out=-118,in=88] (2); 
\draw[->, very thick] (1) -- (3); 
\draw[->, very thick] (2) -- (4); 
\draw[->, very thick] (2) -- (3); 
\draw[->, very thick] (3) -- (4); 
\draw[->, very thick] (3) -- (0); 
\draw[->, very thick] (4) -- (1); 
\draw[->, very thick] (4) -- (0);

\end{tikzpicture}
$$

There are ten perfect matchings $\pi_1,\ldots,\pi_{10}$ of the $\mathrm{dP}_2$ tiling. We show them in figure~\ref{fig:matchings}, where also the corresponding prime toric divisors $D_{\pi_1},\ldots,D_{\pi_{10}}$ are indicated. Here $D_{i}$ denotes the toric divisor defined by the ray generator $v_i$, $0 \leq i \leq 5$. Note that the perfect matchings $\pi_6,\dots,\pi_{10}$ are represented by the toric divisor $D_0$ so that they are internal matchings.

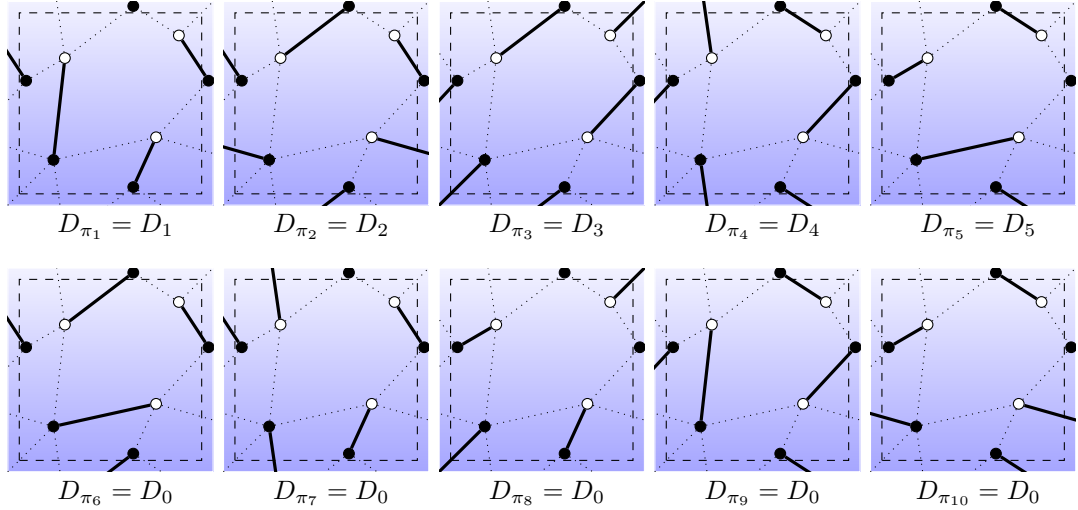
\begin{figure}[t]
\renewcommand{\thesubfigure}{\relax}
\begin{center}
\mbox{
\subfigure[$D_{\pi_1}=D_1$]{\begin{tikzpicture}[scale=0.30,baseline, inner sep=0.5mm]
\clip (-0.5,-0.5) rectangle (8.5,8.5);
\shadedraw[top color=blue!5, bottom color=blue!35, draw=white] (-0.5,-0.5) rectangle (8.5,8.5);

\draw[dashed] (0,0) rectangle (8,8);

\node (b1) at (0.3,5) [circle,draw=black,fill=black] {};
\node (b1s) at (8.3,5) [circle,draw=black,fill=black] {};
\node (b2) at (1.5,1.5) [circle,draw=black,fill=black] {};
\node (b2s) at (9.5,1.5) [circle,draw=black,fill=black] {};
\node (b2ss) at (9.5,9.5) [circle,draw=black,fill=black] {};
\node (b2sss) at (1.5,9.5) [circle,draw=black,fill=black] {};
\node (w1) at (2,6) [circle,draw=black,fill=white] {};
\node (w1ss) at (2,-2) [circle,draw=black,fill=white] {};
\node (w1s) at (10,6) [circle,draw=black,fill=white] {};
\node (b3) at (5,0.3) [circle,draw=black,fill=black] {};
\node (b3s) at (5,8.3) [circle,draw=black,fill=black] {};
\node (w2) at (6,2.5) [circle,draw=black,fill=white] {};
\node (w2s) at (6,10.5) [circle,draw=black,fill=white] {};
\node (w2ss) at (-2,2.5) [circle,draw=black,fill=white] {};
\node (w3) at (7,7) [circle,draw=black,fill=white] {};
\node (w3s) at (-1,7) [circle,draw=black,fill=white] {};
\node (w3ss) at (-1,-1) [circle,draw=black,fill=white] {};
\node (w3sss) at (7,-1) [circle,draw=black,fill=white] {};

\draw[-, very thick] (b2) -- (w1); 
\draw[-, dotted] (b2) -- (w2); 
\draw[-, very thick] (w2) -- (b3); 
\draw[-, dotted] (w2) -- (b1s); 
\draw[-, very thick] (w3) -- (b1s); 
\draw[-, dotted] (w3) -- (b3s); 
\draw[-, dotted] (w1) -- (b3s); 
\draw[-, dotted] (w1) -- (b1); 
\draw[-, dotted] (w2) -- (b2s); 
\draw[-, dotted] (w1s) -- (b1s); 
\draw[-, dotted] (w3) -- (b2ss); 
\draw[-, dotted] (w2s) -- (b3s); 
\draw[-, dotted] (w1) -- (b2sss); 
\draw[-, very thick] (b1) -- (w3s); 
\draw[-, dotted] (b1) -- (w2ss); 
\draw[-, dotted] (b2) -- (w2ss); 
\draw[-, dotted] (b2) -- (w1ss); 
\draw[-, dotted] (b2) -- (w3ss); 
\draw[-, dotted] (b3) -- (w1ss); 
\draw[-, dotted] (b3) -- (w3sss); 
\end{tikzpicture}}

\subfigure[$D_{\pi_2}=D_2$]{\begin{tikzpicture}[scale=.3,baseline, inner sep= 0.5mm]
\clip (-0.5,-0.5) rectangle (8.5,8.5);
\shadedraw[top color=blue!5, bottom color=blue!35, draw=white] (-0.5,-0.5) rectangle (8.5,8.5);

\draw[dashed] (0,0) rectangle (8,8);

\node (b1) at (0.3,5) [circle,draw=black,fill=black] {};
\node (b1s) at (8.3,5) [circle,draw=black,fill=black] {};
\node (b2) at (1.5,1.5) [circle,draw=black,fill=black] {};
\node (b2s) at (9.5,1.5) [circle,draw=black,fill=black] {};
\node (b2ss) at (9.5,9.5) [circle,draw=black,fill=black] {};
\node (b2sss) at (1.5,9.5) [circle,draw=black,fill=black] {};
\node (w1) at (2,6) [circle,draw=black,fill=white] {};
\node (w1ss) at (2,-2) [circle,draw=black,fill=white] {};
\node (w1s) at (10,6) [circle,draw=black,fill=white] {};
\node (b3) at (5,0.3) [circle,draw=black,fill=black] {};
\node (b3s) at (5,8.3) [circle,draw=black,fill=black] {};
\node (w2) at (6,2.5) [circle,draw=black,fill=white] {};
\node (w2s) at (6,10.5) [circle,draw=black,fill=white] {};
\node (w2ss) at (-2,2.5) [circle,draw=black,fill=white] {};
\node (w3) at (7,7) [circle,draw=black,fill=white] {};
\node (w3s) at (-1,7) [circle,draw=black,fill=white] {};
\node (w3ss) at (-1,-1) [circle,draw=black,fill=white] {};
\node (w3sss) at (7,-1) [circle,draw=black,fill=white] {};

\draw[-, dotted] (b2) -- (w1); 
\draw[-, dotted] (b2) -- (w2); 
\draw[-, dotted] (w2) -- (b3); 
\draw[-, dotted] (w2) -- (b1s); 
\draw[-, very thick] (w3) -- (b1s); 
\draw[-, dotted] (w3) -- (b3s); 
\draw[-, very thick] (w1) -- (b3s); 
\draw[-, dotted] (w1) -- (b1); 
\draw[-, very thick] (w2) -- (b2s); 
\draw[-, dotted] (w1s) -- (b1s); 
\draw[-, dotted] (w3) -- (b2ss); 
\draw[-, dotted] (w2s) -- (b3s); 
\draw[-, dotted] (w1) -- (b2sss); 
\draw[-, very thick] (b1) -- (w3s); 
\draw[-, dotted] (b1) -- (w2ss); 
\draw[-, very thick] (b2) -- (w2ss); 
\draw[-, dotted] (b2) -- (w1ss); 
\draw[-, dotted] (b2) -- (w3ss); 
\draw[-, very thick] (b3) -- (w1ss); 
\draw[-, dotted] (b3) -- (w3sss); 

\end{tikzpicture}}

\subfigure[$D_{\pi_3}=D_3$]{\begin{tikzpicture}[scale=.3,baseline, inner sep= 0.5mm]
\clip (-0.5,-0.5) rectangle (8.5,8.5);
\shadedraw[top color=blue!5, bottom color=blue!35, draw=white] (-0.5,-0.5) rectangle (8.5,8.5);

\draw[dashed] (0,0) rectangle (8,8);

\node (b1) at (0.3,5) [circle,draw=black,fill=black] {};
\node (b1s) at (8.3,5) [circle,draw=black,fill=black] {};
\node (b2) at (1.5,1.5) [circle,draw=black,fill=black] {};
\node (b2s) at (9.5,1.5) [circle,draw=black,fill=black] {};
\node (b2ss) at (9.5,9.5) [circle,draw=black,fill=black] {};
\node (b2sss) at (1.5,9.5) [circle,draw=black,fill=black] {};
\node (w1) at (2,6) [circle,draw=black,fill=white] {};
\node (w1ss) at (2,-2) [circle,draw=black,fill=white] {};
\node (w1s) at (10,6) [circle,draw=black,fill=white] {};
\node (b3) at (5,0.3) [circle,draw=black,fill=black] {};
\node (b3s) at (5,8.3) [circle,draw=black,fill=black] {};
\node (w2) at (6,2.5) [circle,draw=black,fill=white] {};
\node (w2s) at (6,10.5) [circle,draw=black,fill=white] {};
\node (w2ss) at (-2,2.5) [circle,draw=black,fill=white] {};
\node (w3) at (7,7) [circle,draw=black,fill=white] {};
\node (w3s) at (-1,7) [circle,draw=black,fill=white] {};
\node (w3ss) at (-1,-1) [circle,draw=black,fill=white] {};
\node (w3sss) at (7,-1) [circle,draw=black,fill=white] {};

\draw[-, dotted] (b2) -- (w1); 
\draw[-, dotted] (b2) -- (w2); 
\draw[-, dotted] (w2) -- (b3); 
\draw[-, very thick] (w2) -- (b1s); 
\draw[-, dotted] (w3) -- (b1s); 
\draw[-, dotted] (w3) -- (b3s); 
\draw[-, very thick] (w1) -- (b3s); 
\draw[-, dotted] (w1) -- (b1); 
\draw[-, dotted] (w2) -- (b2s); 
\draw[-, dotted] (w1s) -- (b1s); 
\draw[-, very thick] (w3) -- (b2ss); 
\draw[-, dotted] (w2s) -- (b3s); 
\draw[-, dotted] (w1) -- (b2sss); 
\draw[-, dotted] (b1) -- (w3s); 
\draw[-, very thick] (b1) -- (w2ss); 
\draw[-, dotted] (b2) -- (w2ss); 
\draw[-, dotted] (b2) -- (w1ss); 
\draw[-, very thick] (b2) -- (w3ss); 
\draw[-, very thick] (b3) -- (w1ss); 
\draw[-, dotted] (b3) -- (w3sss); 

\end{tikzpicture}}

\subfigure[$D_{\pi_4}=D_4$]{\begin{tikzpicture}[scale=.3,baseline, inner sep= 0.5mm]
\clip (-0.5,-0.5) rectangle (8.5,8.5);
\shadedraw[top color=blue!5, bottom color=blue!35, draw=white] (-0.5,-0.5) rectangle (8.5,8.5);

\draw[dashed] (0,0) rectangle (8,8);

\node (b1) at (0.3,5) [circle,draw=black,fill=black] {};
\node (b1s) at (8.3,5) [circle,draw=black,fill=black] {};
\node (b2) at (1.5,1.5) [circle,draw=black,fill=black] {};
\node (b2s) at (9.5,1.5) [circle,draw=black,fill=black] {};
\node (b2ss) at (9.5,9.5) [circle,draw=black,fill=black] {};
\node (b2sss) at (1.5,9.5) [circle,draw=black,fill=black] {};
\node (w1) at (2,6) [circle,draw=black,fill=white] {};
\node (w1ss) at (2,-2) [circle,draw=black,fill=white] {};
\node (w1s) at (10,6) [circle,draw=black,fill=white] {};
\node (b3) at (5,0.3) [circle,draw=black,fill=black] {};
\node (b3s) at (5,8.3) [circle,draw=black,fill=black] {};
\node (w2) at (6,2.5) [circle,draw=black,fill=white] {};
\node (w2s) at (6,10.5) [circle,draw=black,fill=white] {};
\node (w2ss) at (-2,2.5) [circle,draw=black,fill=white] {};
\node (w3) at (7,7) [circle,draw=black,fill=white] {};
\node (w3s) at (-1,7) [circle,draw=black,fill=white] {};
\node (w3ss) at (-1,-1) [circle,draw=black,fill=white] {};
\node (w3sss) at (7,-1) [circle,draw=black,fill=white] {};

\draw[-, dotted] (b2) -- (w1); 
\draw[-, dotted] (b2) -- (w2); 
\draw[-, dotted] (w2) -- (b3); 
\draw[-, very thick] (w2) -- (b1s); 
\draw[-, dotted] (w3) -- (b1s); 
\draw[-, very thick] (w3) -- (b3s); 
\draw[-, dotted] (w1) -- (b3s); 
\draw[-, dotted] (w1) -- (b1); 
\draw[-, dotted] (w2) -- (b2s); 
\draw[-, dotted] (w1s) -- (b1s); 
\draw[-, dotted] (w3) -- (b2ss); 
\draw[-, dotted] (w2s) -- (b3s); 
\draw[-, very thick] (w1) -- (b2sss); 
\draw[-, dotted] (b1) -- (w3s); 
\draw[-, very thick] (b1) -- (w2ss); 
\draw[-, dotted] (b2) -- (w2ss); 
\draw[-, very thick] (b2) -- (w1ss); 
\draw[-, dotted] (b2) -- (w3ss); 
\draw[-, dotted] (b3) -- (w1ss); 
\draw[-, very thick] (b3) -- (w3sss); 

\end{tikzpicture}}

\subfigure[$D_{\pi_5}=D_5$]{\begin{tikzpicture}[scale=.3,baseline, inner sep= 0.5mm]
\clip (-0.5,-0.5) rectangle (8.5,8.5);
\shadedraw[top color=blue!5, bottom color=blue!35, draw=white] (-0.5,-0.5) rectangle (8.5,8.5);

\draw[dashed] (0,0) rectangle (8,8);

\node (b1) at (0.3,5) [circle,draw=black,fill=black] {};
\node (b1s) at (8.3,5) [circle,draw=black,fill=black] {};
\node (b2) at (1.5,1.5) [circle,draw=black,fill=black] {};
\node (b2s) at (9.5,1.5) [circle,draw=black,fill=black] {};
\node (b2ss) at (9.5,9.5) [circle,draw=black,fill=black] {};
\node (b2sss) at (1.5,9.5) [circle,draw=black,fill=black] {};
\node (w1) at (2,6) [circle,draw=black,fill=white] {};
\node (w1ss) at (2,-2) [circle,draw=black,fill=white] {};
\node (w1s) at (10,6) [circle,draw=black,fill=white] {};
\node (b3) at (5,0.3) [circle,draw=black,fill=black] {};
\node (b3s) at (5,8.3) [circle,draw=black,fill=black] {};
\node (w2) at (6,2.5) [circle,draw=black,fill=white] {};
\node (w2s) at (6,10.5) [circle,draw=black,fill=white] {};
\node (w2ss) at (-2,2.5) [circle,draw=black,fill=white] {};
\node (w3) at (7,7) [circle,draw=black,fill=white] {};
\node (w3s) at (-1,7) [circle,draw=black,fill=white] {};
\node (w3ss) at (-1,-1) [circle,draw=black,fill=white] {};
\node (w3sss) at (7,-1) [circle,draw=black,fill=white] {};

\draw[-, dotted] (b2) -- (w1); 
\draw[-, very thick] (b2) -- (w2); 
\draw[-, dotted] (w2) -- (b3); 
\draw[-, dotted] (w2) -- (b1s); 
\draw[-, dotted] (w3) -- (b1s); 
\draw[-, very thick] (w3) -- (b3s); 
\draw[-, dotted] (w1) -- (b3s); 
\draw[-, very thick] (w1) -- (b1); 
\draw[-, dotted] (w2) -- (b2s); 
\draw[-, dotted] (w1s) -- (b1s); 
\draw[-, dotted] (w3) -- (b2ss); 
\draw[-, dotted] (w2s) -- (b3s); 
\draw[-, dotted] (w1) -- (b2sss); 
\draw[-, dotted] (b1) -- (w3s); 
\draw[-, dotted] (b1) -- (w2ss); 
\draw[-, dotted] (b2) -- (w2ss); 
\draw[-, dotted] (b2) -- (w1ss); 
\draw[-, dotted] (b2) -- (w3ss); 
\draw[-, dotted] (b3) -- (w1ss); 
\draw[-, very thick] (b3) -- (w3sss); 

\end{tikzpicture}}
}

\mbox{
\subfigure[$D_{\pi_6}=D_0$]{\begin{tikzpicture}[scale=0.30,baseline, inner sep=0.5mm]
\clip (-0.5,-0.5) rectangle (8.5,8.5);
\shadedraw[top color=blue!5, bottom color=blue!35, draw=white] (-0.5,-0.5) rectangle (8.5,8.5);

\draw[dashed] (0,0) rectangle (8,8);

\node (b1) at (0.3,5) [circle,draw=black,fill=black] {};
\node (b1s) at (8.3,5) [circle,draw=black,fill=black] {};
\node (b2) at (1.5,1.5) [circle,draw=black,fill=black] {};
\node (b2s) at (9.5,1.5) [circle,draw=black,fill=black] {};
\node (b2ss) at (9.5,9.5) [circle,draw=black,fill=black] {};
\node (b2sss) at (1.5,9.5) [circle,draw=black,fill=black] {};
\node (w1) at (2,6) [circle,draw=black,fill=white] {};
\node (w1ss) at (2,-2) [circle,draw=black,fill=white] {};
\node (w1s) at (10,6) [circle,draw=black,fill=white] {};
\node (b3) at (5,0.3) [circle,draw=black,fill=black] {};
\node (b3s) at (5,8.3) [circle,draw=black,fill=black] {};
\node (w2) at (6,2.5) [circle,draw=black,fill=white] {};
\node (w2s) at (6,10.5) [circle,draw=black,fill=white] {};
\node (w2ss) at (-2,2.5) [circle,draw=black,fill=white] {};
\node (w3) at (7,7) [circle,draw=black,fill=white] {};
\node (w3s) at (-1,7) [circle,draw=black,fill=white] {};
\node (w3ss) at (-1,-1) [circle,draw=black,fill=white] {};
\node (w3sss) at (7,-1) [circle,draw=black,fill=white] {};

\draw[-, dotted] (b2) -- (w1); 
\draw[-, very thick] (b2) -- (w2); 
\draw[-, dotted] (w2) -- (b3); 
\draw[-, dotted] (w2) -- (b1s); 
\draw[-, very thick] (w3) -- (b1s); 
\draw[-, dotted] (w3) -- (b3s); 
\draw[-, very thick] (w1) -- (b3s); 
\draw[-, dotted] (w1) -- (b1); 
\draw[-, dotted] (w2) -- (b2s); 
\draw[-, dotted] (w1s) -- (b1s); 
\draw[-, dotted] (w3) -- (b2ss); 
\draw[-, dotted] (w2s) -- (b3s); 
\draw[-, dotted] (w1) -- (b2sss); 
\draw[-, very thick] (b1) -- (w3s); 
\draw[-, dotted] (b1) -- (w2ss); 
\draw[-, dotted] (b2) -- (w2ss); 
\draw[-, dotted] (b2) -- (w1ss); 
\draw[-, dotted] (b2) -- (w3ss); 
\draw[-, very thick] (b3) -- (w1ss); 
\draw[-, dotted] (b3) -- (w3sss); 
\end{tikzpicture}}

\subfigure[$D_{\pi_7}=D_0$]{\begin{tikzpicture}[scale=.3,baseline, inner sep= 0.5mm]
\clip (-0.5,-0.5) rectangle (8.5,8.5);
\shadedraw[top color=blue!5, bottom color=blue!35, draw=white] (-0.5,-0.5) rectangle (8.5,8.5);

\draw[dashed] (0,0) rectangle (8,8);

\node (b1) at (0.3,5) [circle,draw=black,fill=black] {};
\node (b1s) at (8.3,5) [circle,draw=black,fill=black] {};
\node (b2) at (1.5,1.5) [circle,draw=black,fill=black] {};
\node (b2s) at (9.5,1.5) [circle,draw=black,fill=black] {};
\node (b2ss) at (9.5,9.5) [circle,draw=black,fill=black] {};
\node (b2sss) at (1.5,9.5) [circle,draw=black,fill=black] {};
\node (w1) at (2,6) [circle,draw=black,fill=white] {};
\node (w1ss) at (2,-2) [circle,draw=black,fill=white] {};
\node (w1s) at (10,6) [circle,draw=black,fill=white] {};
\node (b3) at (5,0.3) [circle,draw=black,fill=black] {};
\node (b3s) at (5,8.3) [circle,draw=black,fill=black] {};
\node (w2) at (6,2.5) [circle,draw=black,fill=white] {};
\node (w2s) at (6,10.5) [circle,draw=black,fill=white] {};
\node (w2ss) at (-2,2.5) [circle,draw=black,fill=white] {};
\node (w3) at (7,7) [circle,draw=black,fill=white] {};
\node (w3s) at (-1,7) [circle,draw=black,fill=white] {};
\node (w3ss) at (-1,-1) [circle,draw=black,fill=white] {};
\node (w3sss) at (7,-1) [circle,draw=black,fill=white] {};

\draw[-, dotted] (b2) -- (w1); 
\draw[-, dotted] (b2) -- (w2); 
\draw[-, very thick] (w2) -- (b3); 
\draw[-, dotted] (w2) -- (b1s); 
\draw[-, very thick] (w3) -- (b1s); 
\draw[-, dotted] (w3) -- (b3s); 
\draw[-, dotted] (w1) -- (b3s); 
\draw[-, dotted] (w1) -- (b1); 
\draw[-, dotted] (w2) -- (b2s); 
\draw[-, dotted] (w1s) -- (b1s); 
\draw[-, dotted] (w3) -- (b2ss); 
\draw[-, dotted] (w2s) -- (b3s); 
\draw[-, very thick] (w1) -- (b2sss); 
\draw[-, very thick] (b1) -- (w3s); 
\draw[-, dotted] (b1) -- (w2ss); 
\draw[-, dotted] (b2) -- (w2ss); 
\draw[-, very thick] (b2) -- (w1ss); 
\draw[-, dotted] (b2) -- (w3ss); 
\draw[-, dotted] (b3) -- (w1ss); 
\draw[-, dotted] (b3) -- (w3sss); 

\end{tikzpicture}}

\subfigure[$D_{\pi_8}=D_0$]{\begin{tikzpicture}[scale=.3,baseline, inner sep= 0.5mm]
\clip (-0.5,-0.5) rectangle (8.5,8.5);
\shadedraw[top color=blue!5, bottom color=blue!35, draw=white] (-0.5,-0.5) rectangle (8.5,8.5);

\draw[dashed] (0,0) rectangle (8,8);

\node (b1) at (0.3,5) [circle,draw=black,fill=black] {};
\node (b1s) at (8.3,5) [circle,draw=black,fill=black] {};
\node (b2) at (1.5,1.5) [circle,draw=black,fill=black] {};
\node (b2s) at (9.5,1.5) [circle,draw=black,fill=black] {};
\node (b2ss) at (9.5,9.5) [circle,draw=black,fill=black] {};
\node (b2sss) at (1.5,9.5) [circle,draw=black,fill=black] {};
\node (w1) at (2,6) [circle,draw=black,fill=white] {};
\node (w1ss) at (2,-2) [circle,draw=black,fill=white] {};
\node (w1s) at (10,6) [circle,draw=black,fill=white] {};
\node (b3) at (5,0.3) [circle,draw=black,fill=black] {};
\node (b3s) at (5,8.3) [circle,draw=black,fill=black] {};
\node (w2) at (6,2.5) [circle,draw=black,fill=white] {};
\node (w2s) at (6,10.5) [circle,draw=black,fill=white] {};
\node (w2ss) at (-2,2.5) [circle,draw=black,fill=white] {};
\node (w3) at (7,7) [circle,draw=black,fill=white] {};
\node (w3s) at (-1,7) [circle,draw=black,fill=white] {};
\node (w3ss) at (-1,-1) [circle,draw=black,fill=white] {};
\node (w3sss) at (7,-1) [circle,draw=black,fill=white] {};

\draw[-, dotted] (b2) -- (w1); 
\draw[-, dotted] (b2) -- (w2); 
\draw[-, very thick] (w2) -- (b3); 
\draw[-, dotted] (w2) -- (b1s); 
\draw[-, dotted] (w3) -- (b1s); 
\draw[-, dotted] (w3) -- (b3s); 
\draw[-, dotted] (w1) -- (b3s); 
\draw[-, very thick] (w1) -- (b1); 
\draw[-, dotted] (w2) -- (b2s); 
\draw[-, dotted] (w1s) -- (b1s); 
\draw[-, very thick] (w3) -- (b2ss); 
\draw[-, dotted] (w2s) -- (b3s); 
\draw[-, dotted] (w1) -- (b2sss); 
\draw[-, dotted] (b1) -- (w3s); 
\draw[-, dotted] (b1) -- (w2ss); 
\draw[-, dotted] (b2) -- (w2ss); 
\draw[-, dotted] (b2) -- (w1ss); 
\draw[-, very thick] (b2) -- (w3ss); 
\draw[-, dotted] (b3) -- (w1ss); 
\draw[-, dotted] (b3) -- (w3sss); 

\end{tikzpicture}}

\subfigure[$D_{\pi_{9}}=D_0$]{\begin{tikzpicture}[scale=.3,baseline, inner sep= 0.5mm]
\clip (-0.5,-0.5) rectangle (8.5,8.5);
\shadedraw[top color=blue!5, bottom color=blue!35, draw=white] (-0.5,-0.5) rectangle (8.5,8.5);

\draw[dashed] (0,0) rectangle (8,8);

\node (b1) at (0.3,5) [circle,draw=black,fill=black] {};
\node (b1s) at (8.3,5) [circle,draw=black,fill=black] {};
\node (b2) at (1.5,1.5) [circle,draw=black,fill=black] {};
\node (b2s) at (9.5,1.5) [circle,draw=black,fill=black] {};
\node (b2ss) at (9.5,9.5) [circle,draw=black,fill=black] {};
\node (b2sss) at (1.5,9.5) [circle,draw=black,fill=black] {};
\node (w1) at (2,6) [circle,draw=black,fill=white] {};
\node (w1ss) at (2,-2) [circle,draw=black,fill=white] {};
\node (w1s) at (10,6) [circle,draw=black,fill=white] {};
\node (b3) at (5,0.3) [circle,draw=black,fill=black] {};
\node (b3s) at (5,8.3) [circle,draw=black,fill=black] {};
\node (w2) at (6,2.5) [circle,draw=black,fill=white] {};
\node (w2s) at (6,10.5) [circle,draw=black,fill=white] {};
\node (w2ss) at (-2,2.5) [circle,draw=black,fill=white] {};
\node (w3) at (7,7) [circle,draw=black,fill=white] {};
\node (w3s) at (-1,7) [circle,draw=black,fill=white] {};
\node (w3ss) at (-1,-1) [circle,draw=black,fill=white] {};
\node (w3sss) at (7,-1) [circle,draw=black,fill=white] {};

\draw[-, very thick] (b2) -- (w1); 
\draw[-, dotted] (b2) -- (w2); 
\draw[-, dotted] (w2) -- (b3); 
\draw[-, very thick] (w2) -- (b1s); 
\draw[-, dotted] (w3) -- (b1s); 
\draw[-, very thick] (w3) -- (b3s); 
\draw[-, dotted] (w1) -- (b3s); 
\draw[-, dotted] (w1) -- (b1); 
\draw[-, dotted] (w2) -- (b2s); 
\draw[-, dotted] (w1s) -- (b1s); 
\draw[-, dotted] (w3) -- (b2ss); 
\draw[-, dotted] (w2s) -- (b3s); 
\draw[-, dotted] (w1) -- (b2sss); 
\draw[-, dotted] (b1) -- (w3s); 
\draw[-, very thick] (b1) -- (w2ss); 
\draw[-, dotted] (b2) -- (w2ss); 
\draw[-, dotted] (b2) -- (w1ss); 
\draw[-, dotted] (b2) -- (w3ss); 
\draw[-, dotted] (b3) -- (w1ss); 
\draw[-, very thick] (b3) -- (w3sss); 

\end{tikzpicture}}

\subfigure[$D_{\pi_{10}}=D_0$]{\begin{tikzpicture}[scale=.3,baseline, inner sep= 0.5mm]
\clip (-0.5,-0.5) rectangle (8.5,8.5);
\shadedraw[top color=blue!5, bottom color=blue!35, draw=white] (-0.5,-0.5) rectangle (8.5,8.5);

\draw[dashed] (0,0) rectangle (8,8);

\node (b1) at (0.3,5) [circle,draw=black,fill=black] {};
\node (b1s) at (8.3,5) [circle,draw=black,fill=black] {};
\node (b2) at (1.5,1.5) [circle,draw=black,fill=black] {};
\node (b2s) at (9.5,1.5) [circle,draw=black,fill=black] {};
\node (b2ss) at (9.5,9.5) [circle,draw=black,fill=black] {};
\node (b2sss) at (1.5,9.5) [circle,draw=black,fill=black] {};
\node (w1) at (2,6) [circle,draw=black,fill=white] {};
\node (w1ss) at (2,-2) [circle,draw=black,fill=white] {};
\node (w1s) at (10,6) [circle,draw=black,fill=white] {};
\node (b3) at (5,0.3) [circle,draw=black,fill=black] {};
\node (b3s) at (5,8.3) [circle,draw=black,fill=black] {};
\node (w2) at (6,2.5) [circle,draw=black,fill=white] {};
\node (w2s) at (6,10.5) [circle,draw=black,fill=white] {};
\node (w2ss) at (-2,2.5) [circle,draw=black,fill=white] {};
\node (w3) at (7,7) [circle,draw=black,fill=white] {};
\node (w3s) at (-1,7) [circle,draw=black,fill=white] {};
\node (w3ss) at (-1,-1) [circle,draw=black,fill=white] {};
\node (w3sss) at (7,-1) [circle,draw=black,fill=white] {};

\draw[-, dotted] (b2) -- (w1); 
\draw[-, dotted] (b2) -- (w2); 
\draw[-, dotted] (w2) -- (b3); 
\draw[-, dotted] (w2) -- (b1s); 
\draw[-, dotted] (w3) -- (b1s); 
\draw[-, very thick] (w3) -- (b3s); 
\draw[-, dotted] (w1) -- (b3s); 
\draw[-, very thick] (w1) -- (b1); 
\draw[-, very thick] (w2) -- (b2s); 
\draw[-, dotted] (w1s) -- (b1s); 
\draw[-, dotted] (w3) -- (b2ss); 
\draw[-, dotted] (w2s) -- (b3s); 
\draw[-, dotted] (w1) -- (b2sss); 
\draw[-, dotted] (b1) -- (w3s); 
\draw[-, dotted] (b1) -- (w2ss); 
\draw[-, very thick] (b2) -- (w2ss); 
\draw[-, dotted] (b2) -- (w1ss); 
\draw[-, dotted] (b2) -- (w3ss); 
\draw[-, dotted] (b3) -- (w1ss); 
\draw[-, very thick] (b3) -- (w3sss); 

\end{tikzpicture}}
}
\end{center}
\caption{The ten perfect matchings of the $\mathrm{dP}_2$ tiling. The prime toric divisors corresponding to each perfect matching are indicated.} 
\label{fig:matchings} 
\end{figure}

Now let $\mathscr{M}_{\theta}$ be the fine moduli space of $\theta$-stable representations of the algebra $A$ associated to the $\mathrm{dP}_2$ tiling for $\theta=(1,1,1,1,-4)$. If we compute the toric fan of $\mathscr{M}_{\theta}$ as described in \cite{bm0909.2013}, we find that it coincides with that of the total space of the canonical bundle of $\mathrm{dP}_2$. Furthermore, we also verify at once that $\pi_{10}$ is the only $\theta$-stable internal perfect matching. 

To determine the tilting bundle on the canonical bundle of $\mathrm{dP}_2$, we need to fix walks in the periodic quiver that connect the vertex $0$ to all other vertices. Our choice of walks is shown in figure~\ref{fig:walks}. If we choose the trivial walk $\gamma_0$ from~$0$ to itself, the toric divisors computed from the prescription in section~\ref{sec:derivedcategory} are 
\begin{align*}
D(\gamma_0)&=0\, , \\
D(\gamma_1)&= D_1\, ,\\
D(\gamma_2)&= D_1+D_5\, ,\\
D(\gamma_3)&= D_1 + D_4 +D_5\, ,\\
D(\gamma_4)&= D_0+D_1+D_3 +D_4+D_5\, .
\end{align*}
Invoking theorem~\ref{thm:tilting}, we conclude that
\begin{align*}
\Ts = \Os &\oplus \Os(D_1) \oplus \Os( D_1+D_5) \\
& \qquad\oplus \Os(D_1 + D_4 +D_5) \oplus \Os( D_0+D_1+D_3 +D_4+D_5)
\end{align*}
is a tilting object on the canonical bundle of $\mathrm{dP}_2$. 

\begin{figure}
\renewcommand{\thesubfigure}{\relax}
\begin{center}
\subfigure[$\gamma_1\colon 0 \fl 1$]{\begin{tikzpicture}[scale=0.36,baseline, inner sep=0.5mm, >=stealth]
\clip (-0.5,-0.5) rectangle (8.5,8.5);
\shadedraw[top color=blue!5, bottom color=blue!35, draw=white] (-0.5,-0.5) rectangle (8.5,8.5);
\draw[dashed] (0,0) rectangle (8,8);

\node (0) at (0.7,7) [circle,draw=black,fill=white] {3};
\node (0u) at (0.7,-1) [circle,draw=black,fill=white] {3};
\node (0r) at (8.7,7) [circle,draw=black,fill=white] {3};
\node (0ur) at (8.7,-1) [circle,draw=black,fill=white] {3};

\node (1) at (0.7,3.3) [circle,draw=black,fill=white] {0};
\node (1r) at (8.7,3.3) [circle,draw=black,fill=white] {0};

\node (2) at (7,0.8) [circle,draw=black,fill=white] {4};
\node (2o) at (7,8.8) [circle,draw=black,fill=white] {4};
\node (2l) at (-1,0.8) [circle,draw=black,fill=white] {4};
\node (2ol) at (-1,8.8) [circle,draw=black,fill=white] {4};

\node (3) at (3.3,0.8) [circle,draw=black,fill=white] {2};
\node (3o) at (3.3,8.8) [circle,draw=black,fill=white] {2};

\node (4) at (4.7,4.8) [circle,draw=black,fill=white] {1};
\node (4u) at (4.7,-3.2) [circle,draw=black,fill=white] {1};
\node (4l) at (-3.3,4.8) [circle,draw=black,fill=white] {1};

\draw[->, dotted] (0) -- (1); 
\draw[->, very thick] (1) -- (4); 
\draw[->, dotted] (4) -- (3); 
\draw[->, dotted] (3) -- (2); 
\draw[->, dotted] (1r) -- (4); 
\draw[->, dotted] (2o) -- (4); 
\draw[->, dotted] (2) -- (1r); 
\draw[->, dotted] (4) -- (0r); 
\draw[->, dotted] (4) -- (3o); 
\draw[->, dotted] (4u) -- (3); 
\draw[->, dotted] (3) -- (0u); 
\draw[->, dotted] (2) -- (4u); 
\draw[->, dotted] (0ur) -- (2); 
\draw[->, dotted] (1) -- (4l); 
\draw[->, dotted] (2l) -- (1); 
\draw[->, dotted] (3o) -- (0); 
\draw[->, dotted] (4l) -- (0); 
\draw[->, dotted] (0) -- (2ol); 
\draw[->, dotted] (0r) -- (2o); 
\draw[->, dotted] (0u) -- (2l); 

\end{tikzpicture}}
\subfigure[$\gamma_2\colon 0 \fl 2$]{\begin{tikzpicture}[scale=0.36,baseline, inner sep=0.5mm, >=stealth]
\clip (-0.5,-0.5) rectangle (8.5,8.5);
\shadedraw[top color=blue!5, bottom color=blue!35, draw=white] (-0.5,-0.5) rectangle (8.5,8.5);
\draw[dashed] (0,0) rectangle (8,8);

\node (0) at (0.7,7) [circle,draw=black,fill=white] {3};
\node (0u) at (0.7,-1) [circle,draw=black,fill=white] {3};
\node (0r) at (8.7,7) [circle,draw=black,fill=white] {3};
\node (0ur) at (8.7,-1) [circle,draw=black,fill=white] {3};

\node (1) at (0.7,3.3) [circle,draw=black,fill=white] {0};
\node (1r) at (8.7,3.3) [circle,draw=black,fill=white] {0};

\node (2) at (7,0.8) [circle,draw=black,fill=white] {4};
\node (2o) at (7,8.8) [circle,draw=black,fill=white] {4};
\node (2l) at (-1,0.8) [circle,draw=black,fill=white] {4};
\node (2ol) at (-1,8.8) [circle,draw=black,fill=white] {4};

\node (3) at (3.3,0.8) [circle,draw=black,fill=white] {2};
\node (3o) at (3.3,8.8) [circle,draw=black,fill=white] {2};

\node (4) at (4.7,4.8) [circle,draw=black,fill=white] {1};
\node (4u) at (4.7,-3.2) [circle,draw=black,fill=white] {1};
\node (4l) at (-3.3,4.8) [circle,draw=black,fill=white] {1};

\draw[->, dotted] (0) -- (1); 
\draw[->, very thick] (1) -- (4); 
\draw[->, very thick] (4) -- (3); 
\draw[->, dotted] (3) -- (2); 
\draw[->, dotted] (1r) -- (4); 
\draw[->, dotted] (2o) -- (4); 
\draw[->, dotted] (2) -- (1r); 
\draw[->, dotted] (4) -- (0r); 
\draw[->, dotted] (4) -- (3o); 
\draw[->, dotted] (4u) -- (3); 
\draw[->, dotted] (3) -- (0u); 
\draw[->, dotted] (2) -- (4u); 
\draw[->, dotted] (0ur) -- (2); 
\draw[->, dotted] (1) -- (4l); 
\draw[->, dotted] (2l) -- (1); 
\draw[->, dotted] (3o) -- (0); 
\draw[->, dotted] (4l) -- (0); 
\draw[->, dotted] (0) -- (2ol); 
\draw[->, dotted] (0r) -- (2o); 
\draw[->, dotted] (0u) -- (2l); 

\end{tikzpicture}}
\subfigure[$\gamma_3\colon 0 \fl 3$]{\begin{tikzpicture}[scale=0.36,baseline, inner sep=0.5mm, >=stealth]
\clip (-0.5,-0.5) rectangle (8.5,8.5);
\shadedraw[top color=blue!5, bottom color=blue!35, draw=white] (-0.5,-0.5) rectangle (8.5,8.5);
\draw[dashed] (0,0) rectangle (8,8);

\node (0) at (0.7,7) [circle,draw=black,fill=white] {3};
\node (0u) at (0.7,-1) [circle,draw=black,fill=white] {3};
\node (0r) at (8.7,7) [circle,draw=black,fill=white] {3};
\node (0ur) at (8.7,-1) [circle,draw=black,fill=white] {3};

\node (1) at (0.7,3.3) [circle,draw=black,fill=white] {0};
\node (1r) at (8.7,3.3) [circle,draw=black,fill=white] {0};

\node (2) at (7,0.8) [circle,draw=black,fill=white] {4};
\node (2o) at (7,8.8) [circle,draw=black,fill=white] {4};
\node (2l) at (-1,0.8) [circle,draw=black,fill=white] {4};
\node (2ol) at (-1,8.8) [circle,draw=black,fill=white] {4};

\node (3) at (3.3,0.8) [circle,draw=black,fill=white] {2};
\node (3o) at (3.3,8.8) [circle,draw=black,fill=white] {2};

\node (4) at (4.7,4.8) [circle,draw=black,fill=white] {1};
\node (4u) at (4.7,-3.2) [circle,draw=black,fill=white] {1};
\node (4l) at (-3.3,4.8) [circle,draw=black,fill=white] {1};

\draw[->, dotted] (0) -- (1); 
\draw[->, very thick] (1) -- (4); 
\draw[->, very thick] (4) -- (3); 
\draw[->, dotted] (3) -- (2); 
\draw[->, dotted] (1r) -- (4); 
\draw[->, dotted] (2o) -- (4); 
\draw[->, dotted] (2) -- (1r); 
\draw[->, dotted] (4) -- (0r); 
\draw[->, dotted] (4) -- (3o); 
\draw[->, dotted] (4u) -- (3); 
\draw[->, very thick] (3) -- (0u); 
\draw[->, dotted] (2) -- (4u); 
\draw[->, dotted] (0ur) -- (2); 
\draw[->, dotted] (1) -- (4l); 
\draw[->, dotted] (2l) -- (1); 
\draw[->, very thick] (3o) -- (0); 
\draw[->, dotted] (4l) -- (0); 
\draw[->, dotted] (0) -- (2ol); 
\draw[->, dotted] (0r) -- (2o); 
\draw[->, dotted] (0u) -- (2l); 

\end{tikzpicture}}
\subfigure[$\gamma_4\colon 0 \fl 4$]{\begin{tikzpicture}[scale=0.36,baseline, inner sep=0.5mm, >=stealth]
\clip (-0.5,-0.5) rectangle (8.5,8.5);
\shadedraw[top color=blue!5, bottom color=blue!35, draw=white] (-0.5,-0.5) rectangle (8.5,8.5);
\draw[dashed] (0,0) rectangle (8,8);

\node (0) at (0.7,7) [circle,draw=black,fill=white] {3};
\node (0u) at (0.7,-1) [circle,draw=black,fill=white] {3};
\node (0r) at (8.7,7) [circle,draw=black,fill=white] {3};
\node (0ur) at (8.7,-1) [circle,draw=black,fill=white] {3};

\node (1) at (0.7,3.3) [circle,draw=black,fill=white] {0};
\node (1r) at (8.7,3.3) [circle,draw=black,fill=white] {0};

\node (2) at (7,0.8) [circle,draw=black,fill=white] {4};
\node (2o) at (7,8.8) [circle,draw=black,fill=white] {4};
\node (2l) at (-1,0.8) [circle,draw=black,fill=white] {4};
\node (2ol) at (-1,8.8) [circle,draw=black,fill=white] {4};

\node (3) at (3.3,0.8) [circle,draw=black,fill=white] {2};
\node (3o) at (3.3,8.8) [circle,draw=black,fill=white] {2};

\node (4) at (4.7,4.8) [circle,draw=black,fill=white] {1};
\node (4u) at (4.7,-3.2) [circle,draw=black,fill=white] {1};
\node (4l) at (-3.3,4.8) [circle,draw=black,fill=white] {1};

\draw[->, dotted] (0) -- (1); 
\draw[->, very thick] (1) -- (4); 
\draw[->, very thick] (4) -- (3); 
\draw[->, dotted] (3) -- (2); 
\draw[->, dotted] (1r) -- (4); 
\draw[->, dotted] (2o) -- (4); 
\draw[->, dotted] (2) -- (1r); 
\draw[->, dotted] (4) -- (0r); 
\draw[->, dotted] (4) -- (3o); 
\draw[->, dotted] (4u) -- (3); 
\draw[->, very thick] (3) -- (0u); 
\draw[->, dotted] (2) -- (4u); 
\draw[->, very thick] (0ur) -- (2); 
\draw[->, dotted] (1) -- (4l); 
\draw[->, dotted] (2l) -- (1); 
\draw[->, very thick] (3o) -- (0); 
\draw[->, dotted] (4l) -- (0); 
\draw[->, very thick] (0) -- (2ol); 
\draw[->, very thick] (0r) -- (2o); 
\draw[->, very thick] (0u) -- (2l); 

\end{tikzpicture}}
\end{center}
\caption{A set of walks in the periodic quiver.} 
\label{fig:walks} 
\end{figure}
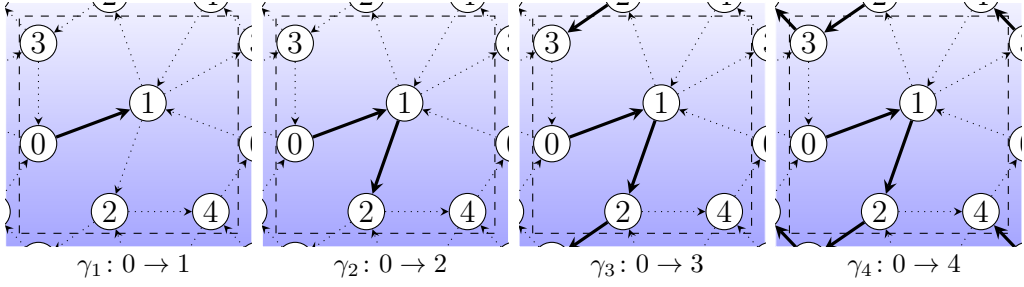

Now, setting $\Es$ to be the restriction of $\Ts$ to $\mathrm{dP}_2$, it is immediate that
\begin{align*}
\Es = \Os \oplus \Os(D_1) \oplus \Os(D_1+D_5) \oplus \Os(D_1 + D_4 +D_5) \oplus \Os(D_1+D_3 +D_4+D_5)
\end{align*}
where, by a mild abuse of notation, $D_i$ denotes the toric divisor in $\mathrm{dP}_2$ associated to the vertex $u_i$, $1 \leq i \leq 5$. Furthermore, by proposition~\ref{prop:tiltingsurface} in appendix~\ref{app:tilting}, $\Es$ is a tilting bundle on $\mathrm{dP}_2$ whose summands are line bundles and hence
$$
(\Os, \Os(D_1), \Os(D_1+D_5), \Os(D_1 + D_4 +D_5), \Os(D_1+D_3 +D_4+D_5))
$$
is a full strong exceptional collection on $\mathrm{dP}_2$. The corresponding quiver is given by
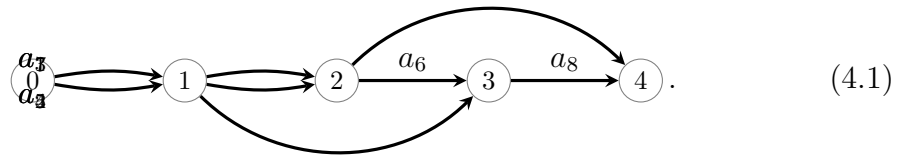
\be\label{acyquiv}
\begin{tikzpicture}[scale=1.0,baseline=-0.1cm, inner sep=1mm, >=stealth]
\node (0) at (0,0) [circle,draw=gray,fill= white] {\footnotesize 0};
\node (1) at (2,0) [circle,draw=gray,fill= white] {\footnotesize 1};
\node (2) at (4,0) [circle,draw=gray,fill= white] {\footnotesize 2};
\node (3) at (6,0) [circle,draw=gray,fill= white] {\footnotesize 3};
\node (4) at (8,0) [circle,draw=gray,fill= white] {\footnotesize 4};
\draw[->, very thick] (0) to [out=10,in=170]  (1) node[midway,above] {$a_{1}$}; 
\draw[->, very thick] (0) to [out=-10,in=190] (1) node[midway,below] {$a_{2}$};
\draw[->, very thick] (1) to [out=315,in=225] (3) node[midway,below] {$a_{5}$}; 
\draw[->, very thick] (2) to [out=45,in=135] (4) node[midway,above] {$a_{7}$}; 
\draw[->, very thick] (1) to [out=10,in=170]  (2) node[midway,above] {$a_{3}$}; 
\draw[->, very thick] (1) to [out=-10,in=190] (2) node[midway,below] {$a_{4}$};
\draw[->, very thick] (2) -- (3) node[midway,above] {$a_{6}$}; 
\draw[->, very thick] (3) -- (4) node[midway,above] {$a_{8}$}; 
\end{tikzpicture} \, .
\ee
If we order the arrows as illustrated above, then the relations\footnote{For a description of how to derive the relations from the collection, see~\cite{MR2464026}.} are
$$
r_1=a_1a_4a_6a_8-a_2a_4a_7\, , \quad r_2=a_2a_5-a_1a_3a_6 \, , \quad r_3=a_3 a_7-a_5a_8 \, .
$$
Thus it follows from~\eqref{Wrel} that the effective superpotential is given by
\begin{align*}
\Wcal_{\infty}&=(\mathbi{a}_1\otimes \mathbi{a}_4\otimes \mathbi{a}_6\otimes \mathbi{a}_8 -\mathbi{a}_2\otimes \mathbi{a}_4\otimes \mathbi{a}_7) \otimes \gbold{\rho}_1 \\
& \qquad + (\mathbi{a}_2\otimes \mathbi{a}_5-\mathbi{a}_1\otimes \mathbi{a}_3\otimes \mathbi{a}_6) \otimes \gbold{\rho}_2+ (\mathbi{a}_3\otimes \mathbi{a}_7-\mathbi{a}_5\otimes \mathbi{a}_8) \otimes \gbold{\rho}_3 \, .
\end{align*}

Now let us pick the internal perfect matching $\pi_{10}$ of figure~\ref{fig:matchings}. Deleting the corresponding arrows in the periodic quiver gives the acyclic quiver:
$$
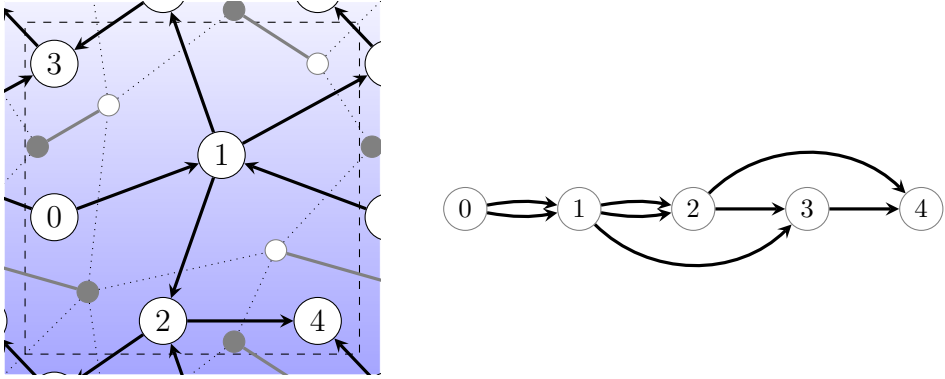
\begin{tikzpicture}[scale=.55, inner sep=1mm, >=stealth]%
\clip (-0.5,-0.5) rectangle (8.5,8.5);
\shadedraw[top color=blue!5, bottom color=blue!35, draw= white] (-0.5,-0.5) rectangle (8.5,8.5);

\draw[dashed] (0,0) rectangle (8,8);

\node (b1) at (0.3,5) [circle,draw=gray,fill= gray] {};
\node (b1s) at (8.3,5) [circle,draw= gray,fill= gray] {};
\node (b2) at (1.5,1.5) [circle,draw= gray,fill= gray] {};
\node (b2s) at (9.5,1.5) [circle,draw= gray,fill= gray] {};
\node (b2ss) at (9.5,9.5) [circle,draw= gray,fill= gray] {};
\node (b2sss) at (1.5,9.5) [circle,draw= gray,fill= gray] {};
\node (w1) at (2,6) [circle,draw= gray,fill=white] {};
\node (w1ss) at (2,-2) [circle,draw= gray,fill=white] {};
\node (w1s) at (10,6) [circle,draw= gray,fill=white] {};
\node (b3) at (5,0.3) [circle,draw= gray,fill= gray] {};
\node (b3s) at (5,8.3) [circle,draw= gray,fill= gray] {};
\node (w2) at (6,2.5) [circle,draw= gray,fill=white] {};
\node (w2s) at (6,10.5) [circle,draw= gray,fill=white] {};
\node (w2ss) at (-2,2.5) [circle,draw= gray,fill=white] {};
\node (w3) at (7,7) [circle,draw= gray,fill=white] {};
\node (w3s) at (-1,7) [circle,draw= gray,fill=white] {};
\node (w3ss) at (-1,-1) [circle,draw= gray,fill=white] {};
\node (w3sss) at (7,-1) [circle,draw= gray,fill=white] {};

\node (0) at (0.7,7) [circle,draw=black,fill=white] {3};
\node (0u) at (0.7,-1) [circle,draw=black,fill=white] {3};
\node (0r) at (8.7,7) [circle,draw=black,fill=white] {3};
\node (0ur) at (8.7,-1) [circle,draw=black,fill=white] {3};

\node (1) at (0.7,3.3) [circle,draw=black,fill=white] {0};
\node (1r) at (8.7,3.3) [circle,draw=black,fill=white] {0};

\node (2) at (7,0.8) [circle,draw=black,fill=white] {4};
\node (2o) at (7,8.8) [circle,draw=black,fill=white] {4};
\node (2l) at (-1,0.8) [circle,draw=black,fill=white] {4};
\node (2ol) at (-1,8.8) [circle,draw=black,fill=white] {4};

\node (3) at (3.3,0.8) [circle,draw=black,fill=white] {2};
\node (3o) at (3.3,8.8) [circle,draw=black,fill=white] {2};

\node (4) at (4.7,4.8) [circle,draw=black,fill=white] {1};
\node (4u) at (4.7,-3.2) [circle,draw=black,fill=white] {1};
\node (4l) at (-3.3,4.8) [circle,draw=black,fill=white] {1};

\draw[->, very thick] (1) -- (4); 
\draw[->, very thick] (4) -- (3); 
\draw[->, very thick] (3) -- (2); 
\draw[->, very thick] (1r) -- (4); 
\draw[->, very thick] (4) -- (0r); 
\draw[->, very thick] (4) -- (3o); 
\draw[->, very thick] (4u) -- (3); 
\draw[->, very thick] (3) -- (0u); 
\draw[->, very thick] (0ur) -- (2); 
\draw[->, very thick] (1) -- (4l); 
\draw[->, very thick] (3o) -- (0); 
\draw[->, very thick] (4l) -- (0); 
\draw[->, very thick] (0) -- (2ol); 
\draw[->, very thick] (0u) -- (2l);
\draw[->, very thick] (0r) -- (2o);

\draw[-, dotted] (b2) -- (w1); 
\draw[-, dotted] (b2) -- (w2); 
\draw[-, dotted] (w2) -- (b3); 
\draw[-, dotted] (w2) -- (b1s); 
\draw[-, dotted] (w3) -- (b1s); 
\draw[-, very thick, gray] (w3) -- (b3s); 
\draw[-, dotted] (w1) -- (b3s); 
\draw[-, very thick, gray] (w1) -- (b1); 
\draw[-, very thick, gray] (w2) -- (b2s); 
\draw[-, dotted] (w1s) -- (b1s); 
\draw[-, dotted] (w3) -- (b2ss); 
\draw[-, dotted] (w2s) -- (b3s); 
\draw[-, dotted] (w1) -- (b2sss); 
\draw[-, dotted] (b1) -- (w3s); 
\draw[-, dotted] (b1) -- (w2ss); 
\draw[-, very thick, gray] (b2) -- (w2ss); 
\draw[-, dotted] (b2) -- (w1ss); 
\draw[-, dotted] (b2) -- (w3ss); 
\draw[-, dotted] (b3) -- (w1ss); 
\draw[-, very thick, gray] (b3) -- (w3sss);  

\end{tikzpicture}
\qquad
\begin{tikzpicture}[scale=1.0,baseline, inner sep=1mm, >=stealth]
\node (0) at (0,2.2) [circle,draw=gray,fill= white] {\footnotesize 0};
\node (1) at (1.5,2.2) [circle,draw=gray,fill= white] {\footnotesize 1};
\node (2) at (3,2.2) [circle,draw=gray,fill= white] {\footnotesize 2};
\node (3) at (4.5,2.2) [circle,draw=gray,fill= white] {\footnotesize 3};
\node (4) at (6,2.2) [circle,draw=gray,fill= white] {\footnotesize 4};
\draw[->, very thick] (0) to [out=10,in=170]  (1); 
\draw[->, very thick] (0) to [out=-10,in=190] (1);
\draw[->, very thick] (1) to [out=315,in=225] (3); 
\draw[->, very thick] (2) to [out=45,in=135] (4); 
\draw[->, very thick] (1) to [out=10,in=170]  (2); 
\draw[->, very thick] (1) to [out=-10,in=190] (2);
\draw[->, very thick] (2) -- (3); 
\draw[->, very thick] (3) -- (4); 

\end{tikzpicture}
$$
As can be inferred from this picture, the quiver obtained from the exceptional collection is the same as the acyclic quiver which is obtained directly from the tiling. If we label the arrows as shown in~\eqref{acyquiv}, then the brane tiling superpotential is
$$
W=(a_1a_4a_6a_8-a_2a_4a_7)\rho_1 +(a_2a_5-a_1a_3a_6)\rho_2+(a_3 a_7-a_5a_8)\rho_3 \, ,
$$
where $\rho_1$, $\rho_2$ and $\rho_3$ are the deleted arrows corresponding to the perfect matching $\pi_{10}$; compare \cite{fhhu0109063, mr0309191}. We readily see that the expressions for $\mathcal{W}_{\infty}$ and $W$ match if we make the right identifications (in particular, the deleted arrows must correspond to relations in the acyclic quiver).

\subsection*{Acknowledgements}

We are grateful to Andres Collinucci, Alastair Craw, Masahiro Futaki, Alastair King, Calin Lazaroiu, Daniel Murfet, Kazushi Ueda and especially Ilka Brunner for helpful discussions, and we thank Ilka Brunner, Andres Collinucci, Alastair Craw and Ingo Runkel for encouragement to write this note.

\appendix

\section[$\Aa$-algebras]{$\boldsymbol{\Aa}$-algebras}\label{app:Ainfinity}

In this appendix we introduce basic notions of the theory of $\Aa$-algebras as well as the specific constructions that we need for our arguments in the main text. For a broader introduction we refer e.\,g.~to~\cite{k9910179, l0310337}. Note that in the main text we denote the $\Aa$-products~$r_{n}$ by~$m_{n}$ to avoid confusion with the relations~$r_{j}$.

An \textit{$\Ainf$-algebra} is a graded vector space $A$ together with linear maps $r_n:A[1]^{\otimes n}\rightarrow A[1]$ of degree $+1$ for all $n\geq 1$ such that
\be\label{defAinf}
\sum_{\genfrac{}{}{0pt}{}{i\geq 0,j\geq 1,}{i+j\leq n}} r_{n-j+1} \circ \left(\id_{A[1]}^{\otimes i}\otimes r_j \otimes\id_{A[1]}^{\otimes (n-i-j)}\right) = 0
\ee
where $A[1]$ denotes the vector space $A$ with the suspended grading, i.\,e.~if $A$ decomposes into its homogeneous components as $A=\bigoplus_iA_i$, then $A[1]_i=A_{i+1}$. 

The first two $A_\infty$-relations in~(\ref{defAinf}) say that $r_{1}$ squares to zero and is a derivative with respect to the multiplication on~$A$ defined via $\mu(a, b):=(-1)^{|a|+1}r_{2}(a, b)$. One may hence consider the cohomology $H_{r_{1}}(A)$ of $r_{1}$ as we will do below. Furthermore, in the case that all higher products vanish, the $n=3$ equation in~\eqref{defAinf} simply says that~$\mu$ is associative. We thus find that a \textit{differential graded algebra} $(A,r_{1},\mu)$ is the special case of an $\Aa$-algebra with $r_{n\geq 3}=0$. 

If one chooses a basis $\{e_{a}\}$ of~$A$ then the (higher) maps $r_{n}$ are determined by their structure constants $Q^{a}_{a_{1}\ldots a_{n}}$ in 
$$
r_{n}(e_{a_{1}},\ldots, e_{a_{n}})=Q^{a}_{a_{1}\ldots a_{n}} e_{a} \, .
$$
When evaluating the $A_\infty$-relations~(\ref{defAinf}) on $e_{a_{1}}\otimes\ldots\otimes e_{a_{n}}$  these relations pick up sign factors according to the Koszul rule, e.\,g.~$(\id\otimes r_1)(a\otimes b)=(-1)^{|a|+1}a\otimes r_1(b)$. This way we find that~(\ref{defAinf}) is equivalently presented in terms of the structure constants $Q^{a}_{a_{1}\ldots a_{n}}$, which have to satisfy the constraints~\eqref{AinfQ}, 
$$
\sum_{\genfrac{}{}{0pt}{}{i\geq 0,j\geq 1,}{i+j\leq n}}(-1)^{|e_{{a_1}}|+\ldots+|e_{a_i}|+i} Q^a_{a_1\ldots a_iba_{i+j+1}\ldots a_n} Q^b_{a_{i+1}\ldots a_{i+j}} = 0 \, .
$$

An $A_\infty$-algebra $(A,r_n)$ is \textit{minimal} iff $r_1=0$. It is \textit{unital} iff there exists $e\in A[1]_{-1}$ such that $r_2(e, a)=-a$, $r_2(a, e)=(-1)^{|a|+1}a$ for all $a\in A[1]$, and all other products $r_n$ vanish if applied to a tensor product involving $e$. $A$ is \textit{cyclic} with respect to a bilinear form $\langle-,-\rangle$ on $A$ iff
\be\label{cyclicity}
\langle a_0,r_n(a_1,\ldots, a_n)\rangle = (-1)^{(|a_0|+1)(|{a}_1|+\ldots+|{a}_n|+n)} \langle a_1,r_n(a_2,\ldots, a_n\otimes a_0)\rangle
\ee
for all homogeneous elements $a_i\in A$.

An \textit{$\Ainf$-morphism} between $\Ainf$-algebras $A$ and $A'$ is a family of linear maps $F_n:A[1]^{\otimes n}\rightarrow A'[1]$ of degree $0$ for all $n\geq 1$ such that
\be\label{defAinfmorph}
\sum_{p=1}^n \sum_{\genfrac{}{}{0pt}{}{1\leq i_1,\ldots,i_p\leq n,}{i_1+\ldots+i_p=n}} r^{A'}_p \circ \left( F_{i_1}\otimes \ldots\otimes F_{i_p}\right)= \sum_{\genfrac{}{}{0pt}{}{i\geq 0,j\geq 1,}{i+j\leq n}}\! F_{n-j+1} \circ \left(\id^{\otimes i}_{A[1]} \otimes r^A_j\otimes\id^{\otimes (n-i-j)}_{A[1]}\right) .
\ee
$(F_n)$ is an \textit{$\Ainf$-isomophism} iff $F_1$ is an isomorphism, and an \textit{$\Ainf$-quasi-isomorphism} iff $F_1$ induces an isomorphism on cohomology with respect to $r_1$. 

Now we can formulate the fundamental theorem in $\Aa$-theory: 

\begin{thm}[\cite{k0504437,m9809,ks0011041}]
Any $\Ainf$-algebra $(A,r_n)$ is $\Ainf$-quasi-isomorphic to a minimal $\Ainf$-algebra. Such a \textit{minimal model} for $A$ is unique up to $\Ainf$-isomorphisms. 
\end{thm}
This means that the cohomology $H_{r_{1}}(A)$ of an $\Aa$-algebra~$A$ can be endowed with an $\Aa$-structure that is compatible with the one on~$A$. In section~\ref{sec:derivedcategory} we use a special version of this construction due to~\cite{kellerainfinrep} in the special case of a differential graded algebra~$A$. 

An \textit{$\Aa$-bimodule} over an $\Aa$-algebra $(A,r_{n})$ is a graded vector space~$M$ together with linear maps $r_{m,n}^M: A[1]^{\otimes m}\otimes M[1]\otimes A[1]^{\otimes n}\rightarrow M[1]$ of degree +1 for all $m,n\geq 1$ such that
\begin{align*}
& \sum_{\genfrac{}{}{0pt}{}{i,j\geq 0,k\geq 1,}{i+j+k=m}} r_{m-k+1,n}^M \circ \left( \id_{A[1]}^{\otimes i} \otimes r_{k} \otimes \id_{A[1]}^{\otimes j} \otimes \id_{M[1]} \otimes \id_{A[1]}^{\otimes m} \right) \\
& \qquad + \sum_{\genfrac{}{}{0pt}{}{i,j,k,l\geq 0,}{i+k=m,j+l=n}} r_{m-i,n-j}^M \circ \left( \id_{A[1]}^{\otimes k} \otimes r_{i,j}^M \otimes \id_{A[1]}^{\otimes l}\right) \\
& \qquad + \sum_{\genfrac{}{}{0pt}{}{i,j\geq 0,k\geq 1,}{i+j+k=n}} r_{m,n-k+1}^M \circ \left( \id_{A[1]}^{\otimes m} \otimes  \id_{M[1]} \otimes \id_{A[1]}^{\otimes i}\otimes r_{k} \otimes \id_{A[1]}^{\otimes j} \right) = 0 \, .
\end{align*}
Any $\Aa$-algebra $(A,r_{n})$ is a bimodule over itself with $r_{m,n}^A=r_{m+1+n}$. There is also a bimodule structure on the dual $A^\vee=\text{Hom}_{\C}(A,\C)$ given by
\begin{align*}
& r_{m,n}^{A^\vee}(a_{1},\ldots,a_{m},\phi,a_{m+1},\ldots,a_{m+n})(a) \\
& \qquad = (-1)^{(\sum_{i=1}^m|a_{i}|)(|\phi|+\sum_{j=1}^n|a_{m+j}|) + |\phi| + \sum_{i=1}^m |a_{i}|\sum_{j=i+1}^m |a_{j}|} \\
& \qquad\qquad \cdot \phi(r_{m+1+n}(a_{m+1},\ldots,a_{m+n},a,a_{1},\ldots,a_{m})) \, .
\end{align*}

When we prove the formula~\eqref{Wrel} in section~\ref{sec:derivedcategory} we need to compute the $\Aa$-structure on $V\oplus V[-3]^\vee$ for the special case of a minimal $\Aa$-algebra~$(V,r_{n})$. This can be done if we know the $\Aa$-structure on the ``trivial extension'' $\Aa$-algebra $A\oplus M$ and the bimodule $M[-1]$ for any $\Aa$-bimodule~$M$ over~$A$, which is given by
\begin{align*}
& r_{n}^{A\oplus M}\big( (a_{1},v_{1}), \ldots, (a_{n},v_{n}) \big) \\
& = \Big(r_{n}(a_{1},\ldots,a_{n}), \sum_{i=1}^n (-1)^{|a_{1}|+\ldots+|a_{i-1}|+i} r_{i-1,n-i}^{M}(a_{1},\ldots,a_{i-1},v_{i},a_{i+1},\ldots,a_{n}) \Big)
\end{align*}
and
\begin{align*}
& r_{m,n}^{M[-1]}(a_{1},\ldots,a_{m},v,a_{m+1},\ldots,a_{m+n}) \\
& \qquad = (-1)^{|a_{1}|+\ldots+|a_{k}|+k} r_{m,n}^{M}(a_{1},\ldots,a_{m},v,a_{m+1},\ldots,a_{m+n}) \, .
\end{align*}

\section[Brane tiling background]{Brane tiling background}\label{app:tiling}

In this appendix, we review some relevant properties concerning brane tiltings. A full treatment of the subject can be found in~\cite{b0901.4662}. 

A graph $\Gamma$ is \textit{bipartite} if its vertex set $V$ can be partitioned into two sets $V'$ and $V''$ in such a way that no two vertices from the same set are adjacent. In fact a graph being bipartite means that the vertices of $\Gamma$ can be coloured black and white such that each edge connects a black vertex to a white one. A \textit{brane tiling} is a polygonal cell decomposition of a torus $T=\R^2/\Z^2$, whose vertices and edges form a bipartite graph $\Gamma$. 

Given a brane tiling $\Gamma \subset T$ one can consider the dual tiling with vertices dual to faces, edges dual to edges and faces dual to vertices. The crucial fact is that the edges of the dual tiling inherit a consistent choice of orientation. Therefore the dual graph is a quiver $Q$, with the additional structure that it provides a tiling of the torus $T$ with oriented faces. We shall refer to the faces of the quiver dual to black (resp.~white) vertices of the brane tiling as black (resp.~white) faces.

As usual, we denote by $Q_0$ the set of vertices and by $Q_1$ the set of arrows of~$Q$. To this information we add the set $Q_2$ of oriented faces. Let $h(a)$ and $t(a)$ denote the head and tail vertices of an arrow $a \in Q_1$. Then there exists a chain complex 
$$
\Z^{Q_2} \stackrel{d_{2}}{\longrightarrow} \Z^{Q_1} \stackrel{d_{1}}{\longrightarrow} \Z^{Q_0} \, ,
$$
where $d_{1}$ is defined by $d_{1}(a)=h(a)-t(a)$, and $d_{2}$ is defined by $d_{2}(f)=\sum_{a \in f} a$. Following~\cite{m0908.3475}, we define the group
$$
G = \Z^{Q_1} / (d_{2}(f)-d_{2}(f') \mid f,f' \in Q_2)
$$
and let $w \colon \Z^{Q_1} \fl G$ be the natural projection. 

Now let $\C Q$ denote the path algebra of $Q$. Then the vector space $\C Q/[\C Q,\C Q]$ consists of all oriented cycles in the quiver $Q$. The consistent orientation of any face $f \in Q_2$ of the quiver means that we may interpret $\partial f$ as an oriented cycle in the quiver. Thus, we can define an element of $W \in \C Q/[\C Q,\C Q]$ by letting
$$
W = \sum_{f \in Q_2} (-1)^f \partial f \,,
$$
where $(-1)^f$ takes value $+1$ on white faces of $Q$, and $-1$ on black faces. This is called the superpotential attached to the brane tiling.

For each arrow $a \in Q_1$ there is a ``derivation'' $\partial_a \colon \C/[\C Q,\C Q] \fl \C Q$ that takes any occurrences of the arrow $a$ in an oriented cycle and removes them leading to a path from $h(a)$ to $t(a)$. Then the superpotential $W$ determines an ideal of relations 
$$
J_W = (\partial_a W\mid a \in Q_1)
$$
in the path algebra $\C Q$; this is usually called the Jacobi ideal. The quotient of the path algebra $\C Q$ by this ideal is the superpotential algebra
\be\label{Jacalg}
A =\C Q/J_W \,.
\ee
It was proved in~\cite{b0901.4662} that under certain consistency conditions on the brane tiling the algebra $A$ is a three-dimensional Calabi-Yau algebra.

A \textit{perfect matching} $\pi$ in a brane tiling $\Gamma \subset T$ is a choice of edges of $\Gamma$ such that each vertex of $\Gamma$ is adjacent to exactly one of the edges. We denote by $P$ the set of all perfect matchings in $\Gamma \subset T$. For any given $\pi \in P$, define the characteristic function $\chi_{\pi}\colon \Z^{Q_1} \fl \Z$ by 
$$
\chi_{\pi}(a)=\left\{
\begin{array}{rl}
1 & \text{if } a\in\pi,\\
0 & \text{if } a  \notin\pi,
\end{array} \right.
$$
where $a\in\pi$ means that~$a$ crosses an edge in~$\pi$. It is plain that $\chi_{\pi}(d_{2}(f))=0$ for all $f \in Q_2$. This implies that $\chi_{\pi}$ induces a well-defined linear map $G \fl \Z$.

\section[A tilting bundle on~$Z$]{A tilting bundle on~$\boldsymbol{Z}$}\label{app:tilting}

Throughout this appendix $Z$ denotes a weak toric Fano surface and $X=\utot(\omega_{Z})$ is the total space of the canonical bundle of $Z$, which we view as a quasi-projective variety with a fibration $\pi \colon X \fl Z$. We shall again let $\iota\colon Z \hookrightarrow X$ denote the inclusion of the zero section. Our aim is to show that the restriction of a tilting bundle on $X$ to the zero section $Z \subset X$ is a tilting bundle. Assertions of this sort have already been considered in \cite[Lem.~7.1]{iu0911.4529}.
 
Before focussing on our specific problem, we record the following useful observation. Let $T$ be a compact object in a compactly generated triangulated category $\Ds$. To say that $\Ds$ is generated by $T$ is to say that the smallest triangulated subcategory of $\Ds$ closed under coproducts that contains~$T$ (and is closed under taking direct summands) coincides with the full subcategory of $\Ds$ consisting of the compact objects. A deep result of~\cite[Thm.~2.1.2]{MR1996800} implies that the latter condition is equivalent to the following one:~for any $E \in \Ds$,
$$
\uHom(T,E[k])=0 \,\,\,\text{for all $k\in \Z$} \quad \Rightarrow \quad E\cong 0 \, .
$$
Lastly, it should be mentioned that the unbounded derived category of quasicoherent sheaves on a separated, Noetherian, regular scheme of finite Krull dimension is compactly generated. In this category a compact object is a perfect complex.

Now let us come back to the situation at hand.

\begin{prop}\label{prop:tiltingsurface}
Let $\Ts$ be a tilting bundle on $X$. Then $\Es = \iota^* \Ts$ is a tilting bundle on $Z$. 
\end{prop}

\begin{proof}
To begin with, we need to ensure that $\Es=\iota^* \Ts$ has no higher self-extensions. The adjoint property of $\Ld \iota^*$ and $\iota_*$ together with the projection formula shows that
\begin{align*}
\Rd \uHom^{\bullet}(\Es,\Es) &= \Rd \uHom^{\bullet}(\Ld \iota^* \Ts,\Ld \iota^* \Ts) \\
            &\cong \Rd \uHom^{\bullet}(\Ts,\iota_*\Ld \iota^* \Ts)\\
            &\cong \Rd \uHom^{\bullet}(\Ts,\Ts \otimes^{\Ld} \iota_* \Os_Z) \, .
\end{align*}
Using the tautological exact sequence
$$
0 \longrightarrow \pi^*\omega_Z^{-1}\longrightarrow \Os_X \longrightarrow \iota_* \Os_Z \, ,
$$
we obtain that the complex $\{ \pi^*\omega_Z^{-1}\fl \Os_X\}$ is a resolution of $\iota_* \Os_Z$ by locally free sheaves. Therefore we find that 
\begin{align*}
\Rd \uHom^{\bullet}(\Es,\Es) &\cong \Rd \uHom^{\bullet}(\Ts,\Ts \otimes^{\Ld} \{ \pi^*\omega_Z^{-1}\fl \Os_X\})\\
&\cong \Rd \uHom^{\bullet}(\Ts,\{\Ts \otimes \pi^*\omega_Z^{-1}\fl \Ts\}) \\
& \cong \Rd \Gamma(\Ts^{\vee} \otimes^{\Ld}\{\Ts \otimes \pi^*\omega_Z^{-1}\fl \Ts\}) \\
& \cong \Rd \Gamma(\{\Ts^{\vee}\otimes\Ts \otimes \pi^*\omega_Z^{-1}\fl \Ts^{\vee}\otimes\Ts\}) \\
& \cong \textstyle\bigoplus_k H^k(X, \{\Ts^{\vee}\otimes\Ts \otimes \pi^*\omega_Z^{-1}\fl \Ts^{\vee}\otimes\Ts\})[-k] \\
& \cong \textstyle\bigoplus_k H^k(X, \Ts^{\vee}\otimes\Ts \otimes \iota_* \Os_Z)[-k]. 
\end{align*}
Since $\Ts$ is a tilting bundle, we have $H^k(X, \Ts^{\vee}\otimes\Ts)=\uExt^k_X(\Ts,\Ts)=0$ for all $k>0$, and thus $H^k(X, \Ts^{\vee}\otimes\Ts \otimes \iota_* \Os_Z)=0$ for all $k>0$, so that indeed
$$
\uExt^k_Z(\Es,\Es)=H^k(\Rd \uHom^{\bullet}(\Es,\Es))=0 \quad \text{for $k > 0$} \, .
$$
 
Finally we have to prove that $\Es$ generates $\Dd^b(\uCoh(Z))$.  The above remarks together with the adjoint property of $\Ld \iota^*$ and $\iota_*$ show that for any object $E \in \Dd^b(\uCoh(Z))$
$$
\uHom(\Es,E[k])=0 \,\,\,\text{for all $k\in \Z$} \quad \Rightarrow \quad \iota_*E\cong 0 \, .
$$
Since $\iota$ is a closed immersion, one obtains $E \cong 0$. The desired conclusion now follows from the foregoing remarks. 
\end{proof}

\end{document}